\documentclass[letterpaper,twocolumn,10pt]{article}
\usepackage{usenix2019_v3}
\usepackage{caption}
\usepackage[noend]{algpseudocode}

\usepackage[linesnumbered,ruled,vlined]{algorithm2e}

\usepackage{amsmath}
\usepackage{amsthm}
\usepackage{amssymb}
\usepackage[normalem]{ulem}
\usepackage{mathtools, nccmath}

\usepackage{caption}
\usepackage{subcaption}

\usepackage{cite}
\usepackage{hyperref}
\usepackage{cleveref}

\usepackage{graphicx}

\usepackage[normalem]{ulem}

\usepackage{booktabs}
\usepackage{xcolor}
\definecolor{forestgreen}{RGB}{34, 139, 34}

\usepackage{bbold}
\usepackage{amsmath,amssymb}
\usepackage{dsfont}

\creflabelformat{equation}{#2\textup{#1}#3}

\newtheorem{definition}{Definition}
\newtheorem*{definition*}{Definition}
\crefname{definition}{Def}{Def}
\crefname{definition*}{Def}{Def}

\newtheorem{proposition}{Proposition}
\newtheorem*{proposition*}{Proposition}
\crefname{proposition}{Prop}{Prop}
\crefname{proposition*}{Prop}{Prop}

\newtheorem{corollary}{Corollary}
\newtheorem*{corollary*}{Corollary}
\crefname{corollary}{Corollary}{Corollary}
\crefname{corollary*}{Corollary}{Corollary}

\newtheorem*{lemma*}{Lemma}
\newtheorem{lemma}{Lemma}
\crefname{lemma}{Lemma}{Lemma}
\crefname{lemma*}{Lemma}{Lemma}
\newtheorem{assumption}{Assumption}
\newtheorem*{assumption*}{Assumption}
\newenvironment{proofsketch}{%
    \proof}{\endproof}

\newtheoremstyle{TheoremNum}
{\topsep}{\topsep}              
{\itshape}                      
{}                              
{\bfseries}                     
{.}                             
{ }                             
{\thmname{#1}\thmnote{ \bfseries #3}}
\theoremstyle{TheoremNum}
\newtheorem{numprop}{Proposition}

\DeclareMathOperator*{\argmax}{argmax}
\usepackage{tikz}
\newcommand*\circled[1]{\tikz[baseline=(char.base)]{
        \node[shape=circle,draw,inner sep=1pt] (char) {#1};}}

\newcommand{\sys}{{NetShaper}}
\newcommand{\nsc}{network side channel}
\newcommand{\nsca}{network side-channel}
\newcommand{\eg}{e.g.,~}
\newcommand{\ie}{i.e.,~}
\newcommand{\medvid}{MedFlix}
\newcommand{\flowmap}{\texttt{FlowMap}}
\newcommand{\prepare}{\texttt{Prepare}}
\newcommand{\ushaper}{\texttt{UShaper}}
\newcommand{\dshaper}{\texttt{DShaper}}

\newcommand{\1}{\mathds{1}}
\newcommand{\istream}{S}
\newcommand{\ostream}{O}

\newcommand{\streamduration}{\tau}
\newcommand{\streamdurationcfg}{\tau}
\newcommand\norm[1]{\|#1\|}
\newcommand{\ssens}{\Delta_{\winlen}}
\newcommand{\qsens}{\Delta_\dpintvl}
\newcommand{\winlen}{W}

\newcommand{\unshapedQ}{Q}
\newcommand{\buffQ}{Q}

\newcommand{\qlen}{L}
\newcommand{\dpintvl}{T}
\newcommand{\qlent}[1]{L_{#1}}
\newcommand{\qlendp}{\tilde{L}}
\newcommand{\qlendpt}[1]{\tilde{L}_{#1}}
\newcommand{\payload}{R}
\newcommand{\dummy}{D}
\newcommand{\numupdates}{\lceil \frac{\winlen}{\dpintvl} \rceil}
\newcommand{\varnumupdates}{N}

\newcommand{\base}{{\bf Base}}
\newcommand{\nsnoshape}{{\bf NS\textsubscript{M}}}
\newcommand{\ns}{{\bf NS}}
\newcommand{\nssim}{{\bf NS\textsubscript{S}}}
\newcommand{\constshape}{{\bf CR}}
\newcommand{\pacer}{{\bf Pacer}}

\newcommand{\citeme}[1]{\textcolor{red} {[{\bf #1}]}}
\newcommand{\am}[1]{\textcolor{cyan}{\bf AM: #1}}
\newcommand{\ml}[1]{\textcolor{orange}{\bf ML: #1}}
\newcommand{\as}[1]{\textcolor{forestgreen}{\bf AS: #1}}

\newcommand{\todo}[1]{{#1}}
\newcommand{\update}[1]{{#1}}

\usepackage{textcomp}

\begin{document}

\title{\Large \bf \sys: A Differentially Private Network Side-Channel Mitigation
System}
\author{Amir Sabzi, Rut Vora, Swati Goswami, Margo Seltzer, Mathias L\'ecuyer,
Aastha Mehta\\
    The University of British Columbia
}
\date{} \thispagestyle{empty}

\maketitle

\begin{abstract}

The widespread adoption of encryption in network protocols has significantly
improved the overall security of many Internet applications. However, these
protocols cannot prevent {\em network side-channel leaks}---leaks of sensitive
information through the sizes and timing of network packets. We present {\sys},
a system that mitigates such leaks based on the principle of traffic shaping.
{\sys}’s traffic shaping provides differential privacy guarantees while adapting
to the prevailing workload and congestion condition, and allows configuring a
tradeoff between privacy guarantees, bandwidth and latency overheads.
Furthermore, {\sys} provides a modular and portable tunnel endpoint design that
can support diverse applications. We present a middlebox-based implementation of
{\sys} and demonstrate its applicability in a video streaming and a web service
application.
\end{abstract}

\section{Introduction}
\label{sec:intro}

With the proliferation of TLS and VPN, traffic encryption has become the de
facto standard for securing data in transit in Internet applications. Traffic
can be encrypted at various layers, such as HTTPS, QUIC, and IPSec.
%
While these protocols prevent {\em direct} data breaches on the Internet, they
cannot~prevent leaks through {\em indirect} observations
of the encrypted traffic.

Indeed, encryption cannot conceal the shape of an
application’s traffic, \ie the sizes, timing, and number of packets sent
and received by an application. In many applications, these parameters strongly
correlate with sensitive information. For instance, 
traffic shape can reveal video
streams \cite{schuster2017beautyburst}, website visits
\cite{wang2014supersequence, bhat2019varcnn},
the content of VoIP conversations \cite{white2011phonotactic}, and
even users’ medical and financial secrets \cite{chen2010reality}.

In such {\em network side-channel leaks}, an adversary (\eg a malicious or
a compromised ISP) observes the shape of an application's traffic as it
passes through a link under its control and
infers the application's sensitive data from this shape.


Obfuscation techniques, which add ad hoc noise~\cite{luo2011httpos}
or adversarial noise \cite{shan2021dolos, nasr2021defeating, rahman2020mockingbird,
hou2020wf, abusnaina2020dfd, gong2022surakav}
in an application's network traces, do not provide comprehensive protection
against {\nsca} attacks \cite{zhang2019statistical}.
In fact, recent advances in machine learning (ML)
have greatly improved the ability to filter
out noise due to congestion or path variations and infer secrets from noisy data
\cite{schuster2017beautyburst, bhat2019varcnn, hayes2016kfp, sirinam2018df}. For
instance, our own novel classifier based on Temporal Convolution
Networks (TCN) \cite{bai2018tcn} can infer video streams even from short bursts
of noisy
measurements over the Internet (see \S\ref{sec:background}~for details).
Alternatively, a sensitive application could try to improve side-channel
resilience by splitting traffic over multiple network paths
\cite{cadena2020trafficsliver, wang2022leveraging} or by
using dedicated physical links that are not controlled by the adversary.
However, such solutions are inadequate against a powerful adversary
that can monitor a large fraction of the Internet
\cite{beckerle2022splitting} and may incur prohibitive network
administration costs for small users on the Internet.


In contrast, a principled and practical approach to mitigating network
side-channel leaks is {\em traffic shaping}. It involves modifying the victim’s
packet sizes and timing to make the resultant shape independent of secrets, so
that an adversary cannot infer the secrets despite observing the (shaped)
traffic.

Constant shaping involves sending fixed-sized packets at a constant rate, which
is secure but incurs non-trivial bandwidth and/or latency overhead for
applications with variable or bursty workloads \cite{saponas2007devices}.
Variable shaping strategies attempt to adapt traffic shapes to reduce the
overhead
at the cost of some privacy. However, the state-of-the-art (SOTA) variable
shaping strategies rely on ad hoc heuristics that yield weak privacy guarantees
\cite{wang2014supersequence, nithyanand2014glove, wang2017walkietalkie} or
unbounded privacy leaks \cite{gong2020zero, cai2014csbuflo,
lu2018dynaflow, juarez2016wtfpad, cai2014tamaraw}.
Some techniques provide strong guarantees but
require extensive a priori profiling of an applications' traffic to compute
shapes
\cite{mehta2022pacer, zhang2019statistical}.

%
%
%
In addition to a shaping strategy, network side-channel mitigation also
requires a robust implementation of packet padding and transmit scheduling.
Many solutions attempt to protect traffic by controlling shaping from only one
end of a communication (\ie either a client or a server) and provide only
best-effort protection
\cite{luo2011httpos, smith2022qcsd, cherubin2017llama}. Other solutions rely on
trusting third-party mediators (\eg Tor bridges), which implement shaping
between the clients and mediators and between the servers and mediators
\cite{mohajeri2012skypemorph, winter2013scramblesuit}.
In Pacer \cite{mehta2022pacer}, both application endpoints integrate a
shaping system to comprehensively mitigate network side channels. However, Pacer
encumbers application end hosts with non-trivial changes in the network
stack to implement shaping, thus deterring adoption.

In this work, we address two main questions. First, is there a variable traffic
shaping strategy that provides quantifiable and tunable privacy guarantees at
runtime without requiring extensive pre-profiling of application traffic?
Second, can traffic shaping be provided as a generic, portable, and efficient
solution that can be integrated in different network settings and can support
diverse applications?

We present {\sys}, a network side-channel mitigation system that answers both
questions in the affirmative.
First, {\sys} relies on a differential privacy (DP) based traffic shaping
strategy, which provides quantifiable and tunable privacy guarantees.
{\sys} specifies the privacy parameters for a configurable window of
transmission. Moreover, it can configure the parameters independently for each
direction of traffic on a communication link. The DP guarantees can be composed
based on these parameters to achieve bounded privacy leaks for arbitrary
bidirectional traffic.
Overall, applications can tune the shaping based only on the privacy guarantees
they desire and the overheads they can afford, without the need for
profiling their traffic.
While strong privacy guarantees require DP parameters that incur
large overheads, in practice, {\sys} can defeat SOTA attacks with even small
amounts of DP noise, thus incurring low overheads.

Secondly, we present a traffic shaping tunnel with a modular endpoint design
that can conceptually be integrated with any network stack and within any node.
The tunnel implements padding and transmit scheduling of packets while
adhering to the DP guarantees {\em by design}.
{By placing the tunnel endpoint in a middlebox at the edge of the
private network, {\sys} can simultaneously protect
the traffic of multiple applications. Moreover, the middlebox can amortize the
shaping overheads among multiple flows without compromising the privacy for
individual flows.
}

Together, the DP shaping strategy and {\sys}'s tunnel provide effective {\nsca}
mitigation for diverse applications, such as video streaming and web services,
with modest overheads.
To the best of our knowledge, {\sys} is the first system to provide dynamic
traffic shaping with quantifiable and tunable privacy guarantees based on DP.


\textbf{Contributions.}
(i)
We design a new attack classifier based on a
Temporal Convolution Network (TCN) \cite{bai2018tcn}
and demonstrate its ability to infer
videos streams from traffic shapes under noisy network traffic measurements in
the Internet (\S\ref{sec:background}).
(ii) We model network side-channel mitigations as a differential privacy problem
and provide a traffic shaping strategy that offers
($\epsilon$,$\delta$)-differential privacy guarantees (\S\ref{sec:dp}).
(iii) We design a QUIC-based traffic shaping tunnel and present a
middlebox-based implementation of the tunnel, which supports traffic shaping
while adhering to DP guarantees (\S\ref{sec:design}).
(iv) We demonstrate {\sys}'s efficacy in defeating a SOTA classifier
\cite{schuster2017beautyburst} and our new TCN classifier.
We empirically evaluate the tradeoffs between {\sys}'s privacy guarantees
and performance overheads while mitigating {\nsca} leaks in two classes of
applications that have already been used in prior work, namely video streaming
and web service (\S\ref{sec:eval}).
(vi) We present a formal proof of {\sys}'s differential privacy
guarantees~(\S\ref{appendix:dp}).

\section{Background, Motivation, and Overview}
\label{sec:background}


\subsection{Network Side-Channel Attacks}
\label{subsec:attack-bg}

We start by explaining the workings of a network side-channel attack with an
example application.
Consider {\medvid}, a fictitious medical video service that offers videos on
symptoms, treatment procedures, and post-operative care.
The goal of an adversary is to infer the videos streamed by users visiting the
service.
ISPs can aggregate such information to build per-user profiles and subsequently
monetize them.
Additionally, competitors might exploit {\nsc}s to acquire corporate intelligence
without detection.

The adversary performs the attack in two stages. First, the adversary requests
each video from the medical service as a client and collects traces of the
bidirectional network traffic generated while streaming each video. The
adversary may collect multiple traces for each video stream to account for
network variations in the traffic. The adversary then builds a classifier over
the captured traces to identify the video streams.
Prior work has used several features for classification, such as packet sizes,
inter-packet timing, total bytes transferred in a burst of packets, the burst
duration and inter-burst interval, and the direction of packets or bursts
\cite{schuster2017beautyburst}.

We reproduce the Beauty and the Burst (BB) classifier
\cite{schuster2017beautyburst}, a state-of-the-art CNN
classifier for classifying video streams from network traces.
Furthermore, we present a new TCN classifier \cite{bai2018tcn}, which is an
improvement over the BB classifier.
We describe the classifiers in \S\ref{appendix:tcn}.
Here, we evaluate~the efficacy of the two classifiers for a network side-channel
attack.

We set up a video service and a video client as two Amazon AWS VMs placed in
Oregon and Montreal, respectively. The video server hosts a dataset of
\update{100 YouTube videos} at 720p resolution with MPEG-DASH encoding
\cite{mpeg-dash}. The client streams the first \update{5 min} of each video
\update{over HTTPS} and collects the resulting network packet traces using
tcpdump.
We stream each video \update{100} times, thus collecting a total of
\update{10,000} traces.

%
%
The classifiers' goal is to predict the video from a network trace.
%
For each classifier, we transform each packet trace into a sequence of burst
sizes transmitted within 1s windows~and normalize the sequence by dividing
each burst size by the total size of all bursts.
%
We evaluate the performance of both classifiers with three datasets: a small
dataset consisting of 20 videos with their 100 traces each (\ie total
2000 traces), a medium
dataset with 40 videos (4000 traces),  and a large dataset comprising
all \update{100 videos} (\update{10000} traces).
We train the classifiers for 1000 epochs with an 80-20 train-test split.
%
BB's classification accuracy, recall, and precision with the small dataset are
0.85,~0.85, and 0.78, respectively, which drop to
0.61,~0.63, and 0.49, respectively, with the medium dataset, and further
drop to 0.01 each for the large dataset.~TCN's accuracy, recall, and precision
are above 0.99 for all datasets.~TCN performs better than BB because it
is a complex model with residual layers and, hence, is robust to noise in the
traces.
%

%
%

Similarly, advanced ML classifiers are capable of identifying web traffic
\cite{bhat2019varcnn, sirinam2018df}.
%
In general, classifiers will continue to evolve,
increasing the adversary's capabilities to make inferences from noisy
measurements.
Hence, we need principled mitigations that address
current SOTA attacks and achieve quantifiable
leakage, which can be configured based on privacy
requirements and overhead tolerance.

\subsection{Key Ideas}
\label{subsec:key-ideas}

A secure and practical {\nsca} mitigation system must satisfy the following
design goals:
{\bf G1.} Mitigate leaks through all aspects of the shape of transmitted traffic,
{\bf G2.} Provide quantifiable and tunable privacy guarantees for the
communication parties,
{\bf G3.} Minimize overheads incurred while guaranteeing privacy,
{\bf G4.} Support a broad class of applications, and
{\bf G5.} Require minimal changes to applications.

{\sys}'s shaping prevents leaks of the traffic content through sizes and
timing of packets transmitted along each direction between application nodes
(G1).
In addition, {\sys} relies on the following three key ideas.
%

\textbf{Differentially private shaping.}
Unlike constant shaping, variable shaping can adapt traffic shape, potentially
based on runtime workload patterns, and thus significantly reduce shaping
overheads.
{Unfortunately, existing variable shaping techniques either have unbounded
privacy leaks, offer only weak privacy guarantees, or require extensive
profiling of an application's network traces.
}
%
{\sys}'s novel differential privacy (DP) based shaping strategy provides
quantifiable and tunable bounds on privacy leaks, without relying on profiling
of application traffic (G2).
{\sys} shapes an application's traffic in periodic, fixed-length {\em
shaping intervals} and provides DP in the length of the application byte stream
(burst) accumulated within each interval.
The DP guarantees compose over a sequence of multiple intervals and, thus, for
streams of arbitrary length (albeit with
degraded guarantees).

\textbf{Shaping in a middlebox.}
{\sys} uses a tunnel abstraction to implement traffic shaping.
The tunnel shapes application traffic such that an
adversary observing the tunnel traffic cannot infer application secrets.
In principle, a tunnel endpoint could be integrated with the application host
(\eg in a VM isolated from the end-host application) or in a
separate node through which the application's traffic passes.
{\sys} relies on the second approach and implements the tunnel endpoint as a
middlebox, which could be integrated with an existing network element, such as a
router, a VPN gateway, or a firewall.
The middlebox implementation enables securing multiple applications without
requiring modifications on individual end hosts (G4).
Furthermore, it allows pooling multiple flows with the same privacy requirements
in the same tunnel, which helps to amortize the per-flow overhead (G3).

\textbf{Minimal modifications to end applications.}
By default, {\sys} shapes all traffic through a tunnel with a fixed differential
privacy guarantee.
However, an application can explicitly specify different DP parameters to adapt
the privacy guarantee enforced for its traffic,
as well as bandwidth and latency constraints and any prioritization preferences
on a per-flow basis.
This requires only a small modification in the
application; it must transmit a shaping configuration message
to the middlebox.
Thus, {\sys} offers a balance between being fully
application-agnostic and optimizing for privacy or overhead with minimal
support from applications (G2, G5).

\subsection{Threat Model}
\label{subsec:threat-model}

{\sys}'s goal is to hide the content of an application's network traffic.
Hiding the type of traffic \cite{shapira2019flowpic}, the
communication protocol \cite{winter2013scramblesuit}, or the
application identity \cite{dyer2012peekaboo, danezis2009https} are
non-goals, although {\sys} can adapt its shaping strategy to address these goals.
The applications are non-malicious and do not leak their own secrets.

We assume that the application endpoints are inside separate trusted
private networks (\eg each node is behind a VPN gateway node) and the adversary
cannot infiltrate the private network, or the clients and servers within it.
(Thus, we exclude covert attacks~\cite{zhang2011predinteractive} and
colocation-based~attacks~\cite{schuster2017beautyburst,mehta2022pacer}).
The adversary controls network links in the public Internet (\eg ISPs) and can
record, measure, and tamper with the victim application's traffic as it
traverses the links under the adversary's control.
The adversary can precisely record the traffic shape---the sizes, timing,
and direction of packets---between the {gateway nodes}.
In particular, it may have access to observations of arbitrary known
streams to train its attack.
It may also have knowledge about {\sys}, including its
shaping strategy and privacy configurations.

We do not consider threats due to observing the IP addresses of
packets~\cite{hoang2021domain}, although {\sys} can hide IP addresses of
applications behind a shared traffic shaping tunnel.


{\sys} does not address leaks of one application's sensitive data through the
traffic shape of colocated benign applications.
Such leaks can arise, for instance, due to microarchitectural interference among
applications colocated on a host or among their flows if they pass through
shared links.
End hosts could implement orthogonal mitigations against colocated
applications \cite{mehta2022pacer,
shi2011limiting,
varadarajan2014scheduler, braun2015robust}
and combine
{\sys}'s shaping with TDMA scheduling on network
links~\cite{beams2021ifs, vattikonda2012tdma}.

We present a middlebox-based {\sys} implementation
that can be integrated with an organization's trusted gateway router.
%
{\sys}'s trusted computing base (TCB) includes all components in the
organization's private network and the middleboxes.
Bugs, vulnerabilities or side channels in the middleboxes that threaten traffic
confidentiality could be mitigated using orthogonal
techniques, such as software fault isolation~\cite{tan2017sfi}, resource
partitioning \cite{liu2016catalyst}, and constant-time implementation techniques
\cite{
    almeida2016verifying}.

Under these assumptions, {\sys} prevents leaks of application secrets
through the sizes and timing of packets transmitted in
either direction between the application~endpoints.
\subsection{A Primer on Differential Privacy}
\label{subsec:DP-background}

{Finally, we provide a brief primer on DP, the steps
involved in building a DP mechanism, and the key properties of DP that are
relevant in the context of traffic shaping (\S\ref{sec:dp}).}

Developed originally for databases, DP is a technique to provide aggregate
results without revealing information about individual database
records.
Formally, a randomized algorithm $\mathcal{M}$ is $(\varepsilon, \delta)$-DP
if, for all $\mathcal{R} \subseteq \textrm{Range}(\mathcal{M})$ and for all
{\em neighboring} databases $d, d'$ that differ in only one element:
\begin{align}
\label{eq:dp}
P[\mathcal{M}(d)\in \mathcal{R}] \leq e^{\varepsilon}~P[\mathcal{M}(d') \in
\mathcal{R}] + \delta
\end{align}
The parameter $\varepsilon$ represents the {\em privacy loss} of algorithm
$\mathcal{M}$, \ie given a
result of $\mathcal{M}$, the information gain for any adversary on learning
whether the input database is $d$ or $d'$ is at most $e^{\varepsilon}$
\cite{kasiviswanathan2014semantics}.
The $\delta$ is the probability with which $\mathcal{M}$ fails to bound the
privacy loss to $e^{\varepsilon}$.

Building such a randomized DP algorithm $\mathcal{M}$ involves three
main steps: (i) defining neighboring database states, (ii) defining a database
query and determining the sensitivity of the query to changes in neighboring
databases, and (iii) adding noise to the query.
Neighboring databases $d$ and $d'$, as mentioned above, are characterized by the
{\em distance} between the databases, which quantifies the granularity at which
the DP guarantee applies.
Traditionally, this distance is defined as the number of
records that differ between $d$ and $d'$. However, DP also extends to other
neighboring definitions and distance metrics used in specific settings
\cite{chatzikokolakis2013broadening, lecuyer2019certified}.

Given a database query $q$, the sensitivity of the query $\Delta q$ is the max
difference in the result achieved when the query is executed on the neighboring
databases $d$ and $d'$. Intuitively, a larger $\Delta q$ implies higher
probability of an adversary inferring from a result the database on which the
query was executed, thus incurring higher privacy loss.
To mitigate this privacy loss, a DP mechanism therefore adds noise to
the query result to hide the true result and the underlying database. Popular
noise mechanisms are Laplace \cite[\S3.3]{dwork2014algorithmic} and Gaussian
\cite{dong2022gaussian}.
DP provides three properties.
As we will show in \S\ref{subsec:dp-queue-measurements}, these are also of
relevance to {\sys}'s DP traffic shaping.
First, DP is resilient to post-processing: given the result $r$ of any
$(\varepsilon, \delta)$-DP mechanism $\mathcal{M}$, any function $f(r)$
of the result is also $(\varepsilon, \delta)$-DP.
As a result, any computation or decision made on a DP result is still DP with
the same guarantees.
Second, DP is closed under adaptive sequential {\em composition}: the combined
result of two DP mechanisms $\mathcal{M}_1$ and $\mathcal{M}_2$ is also DP,
though with higher losses ($\varepsilon$ and $\delta$).
We use the R\'enyi-DP definition~\cite{mironov2017renyi} to achieve simple but
strong composition results and subsequently convert the results
back to the standard DP definition.
%
Third, DP is robust to auxiliary information: the guarantee from \Cref{eq:dp}
holds regardless~of any side information known to an attacker.
Therefore, the attacker's knowledge of the shaping mechanism does not affect its
privacy guarantees.
That is, an attacker knowing or controlling part of the database cannot extract
more knowledge from a DP result than without this side information.


\section{Differentially Private Traffic Shaping}
\label{sec:dp}

\begin{figure}[t]
    \centering
    \includegraphics[width=\columnwidth]{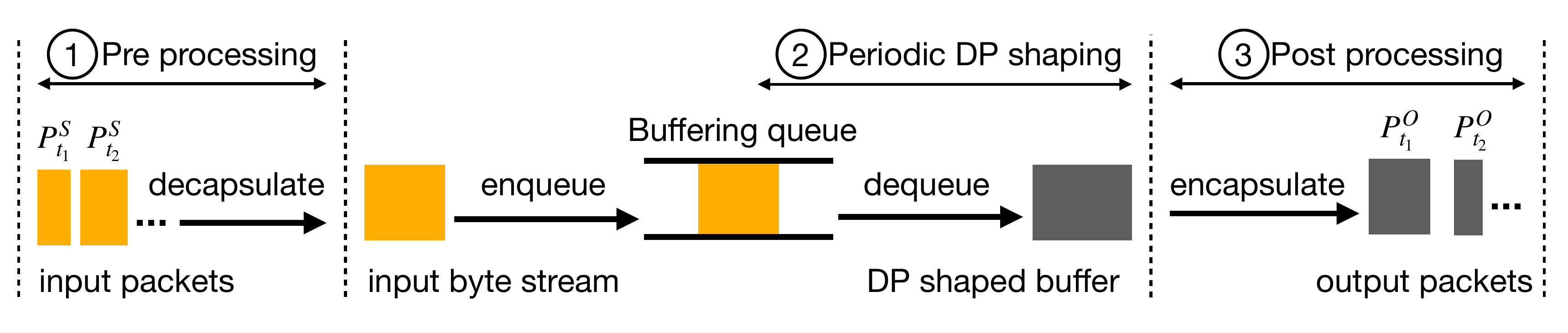}
    \caption{Overview of DP shaping}
    \vspace{-0.4cm}
    \label{fig:dp-overview}
\end{figure}

The goal of differentially private shaping is to dynamically adjust packet
sizes and timing based on the available data stream, while ensuring that the DP
guarantees hold for any information that an adversary
(\S\ref{subsec:threat-model}) can observe.

%
We first formally model an application's input stream as a packet sequence
$\istream = \{{P^{\istream}_1}, {P^{\istream}_2}, {P^{\istream}_3}, \dots \}$,
where ${P_i}^\istream = (l^{\istream}_i, t^{\istream}_i)$ indicates that the
$i$\textsuperscript{th} input packet in $\istream$ has length $l^{\istream}_i$
bytes and is transmitted at timestamp $t^{\istream}_i$.
We call the total duration of a finite stream $\streamduration^S
\triangleq \max t^{\istream}_i - t^{\istream}_1$.
Without shaping, an adversary can precisely observe $\istream$ and infer
the content, which is correlated with the~stream (\S\ref{subsec:attack-bg}).

\Cref{fig:dp-overview} provides a high-level overview of the
differentially-private shaping mechanism. Shaping happens in periodic intervals
of fixed length $\dpintvl$, called the {\em DP shaping intervals}.
\circled{1} As the packets in an application's stream arrive, the payload bytes
extracted from the packets are placed into a {\em buffering queue}, $\buffQ$.
\circled{2} In each periodic interval, the DP shaping algorithm
performs a {\em DP query}: it measures the length of $\buffQ$ with DP,
to determine the number of bytes to transmit in the next interval.
{\sys} then prepares a DP shaped buffer using the
payload bytes from $\buffQ$, and additional dummy bytes if required.
This shaped buffer is enqueued to be sent over the network at the end of
the DP shaping interval, right before the next interval starts.
\circled{3} Finally, data in the shaped buffer may be split into one or more
packets, as part of a post-processing step, and transmitted to the network.

The size of each shaped buffer generated in an interval has ($\varepsilon_{T},
\delta_{T}$)-DP guarantees. The per-interval guarantees
compose over a sequence of multiple intervals, thus providing DP guarantees for
traffic streams of arbitrary lengths. The guarantees degrade as the stream length increases (\Cref{prop:dp})

%
We now discuss the steps of building a DP mechanism for
traffic streams in \S\ref{subsec:building-blocks} and the
privacy guarantees in \S\ref{subsec:dp-queue-measurements}.

\subsection{DP for Traffic Streams}
\label{subsec:building-blocks}

We now discuss the three steps for building a DP mechanism on
traffic streams. Specifically, we present our neighboring definition for
streams, define the DP query on streams that {\sys} runs and bounds its sensitivity,
and show the DP mechanism that we use to make the query DP.



\textbf{Step 1: Neighboring Definition.}
Building a DP mechanism requires a suitable definition of neighboring
streams, which requires a notion of distance between streams.
The neighboring definition has two implications. First, it
determines the sensitivity for a ``query'' subject to the DP mechanism and, in
turn, the amount of noise necessary for ensuring a given DP guarantee. Second,
it defines the granularity of privacy guarantee: intuitively, neighboring streams
are indistinguishable based on only the results of the DP mechanism applied to
them.
Hence, we aim for a neighboring definition for which ``streams have many
neighbors'', and the sensitivity of the DP query we want to run is as low as
possible.

\update{Defining neighbors over streams has two challenges. First, the streams
can be arbitrarily long and the distance between streams would typically
increase with stream length. Secondly, if we consider streams with packet
timestamps at the finest granularity, the distance between the streams may be as
large as the sum of sizes of all packets in both streams. Both these factors
imply that a stream would either have few neighbors (enabling weak privacy), or
the neighboring definition would need to specify a large distance threshold,
requiring a lot of noise for strong DP guarantees.}

To keep the neighboring distance threshold small, we take two steps.
First, we define the notion of a {\em neighboring window}, which is a time
interval of configurable length $\winlen$ over input streams.
For notational convenience, we set $\winlen$ as an integral multiple of the DP
shaping interval $\dpintvl$ and write
$\winlen = k \dpintvl$, although this is not a strict requirement.

\update{Secondly, to measure the distance between streams over a neighboring
window, we consider the total bytes in each stream (burst lengths) transmitted
within coarse-grained time intervals and use the difference in the burst lengths
in the intervals within the window. Intuitively, we need an interval granularity
that is coarse enough to generate similar burst length sequences in more
streams, but is also small enough to bound the accumulation of differences over
time (to bound sensitivity, our next step).}
As will become clear with \Cref{prop:sensitivity}, the DP shaping
interval $\dpintvl$ is the coarsest granularity that we can use.
Hence, considering a neighboring window starting at a timestamp $t_w$, \ie
$[t_w,~t_w + \winlen)$,
we define the following representation of stream $\istream$ over the window
$[t_w,~t_w +\winlen)$ at granularity $\dpintvl$: $S_{t_w, W} = \{L^S_{t_w},
L^S_{t_w+\dpintvl}, L^S_{t_w+2\dpintvl}, \ldots, L^S_{t_w+(k-1)\dpintvl}\}$,
where $L^S_{t} \triangleq \sum_{(l^S_i, t^S_i) \in S} \1\{t^S_i \in [t,
t+\dpintvl)\} l^{\istream}_i$ is the total application bytes accumulated in
$\buffQ$ in the interval $[t, t+\dpintvl)$.

We now present the following neighboring definition:
\begin{definition}
\label{def:neighboring-streams}
Two streams $S$ and $S'$ are neighbors if, for any neighboring window
$[t_w,~t_w +\winlen)$, the L1-norm distance between their representations
$S_{t_w, W}$ and  $S'_{t_w, W}$  is less than $\ssens$.
Formally, $S$ and $S'$ are neighbors if:
\begin{equation}
\max_{t_w} {\norm{S_{t_w, W} - S'_{t_w, W}}}_1 \leq \ssens .
\end{equation}
\end{definition}
We utilize the L1-norm (the sum of absolute values) to quantify the distance
between two traffic streams, as it captures differences in both traffic size and
temporal alignment, at the granularity of $\dpintvl$.
%
%
We will show (in \Cref{prop:sensitivity}) that despite our restriction
of defining neighboring streams at a granularity of $\dpintvl$, and of computing
distances over windows of length $\winlen$, we can quantify {\sys}'s DP
guarantees for streams at any granularity and of any length.

The neighboring window length $\winlen$ and neighboring distance
$\ssens$ are both configuration parameters, which are set before the start of an
application's transmission.
In practice, $\winlen$ would be in the order of milliseconds to seconds.
Subsequently, one would determine $\ssens$ based on the typical
difference of traffic between application streams over windows of length
$\winlen$.
The practical upper bound for $\ssens$ is the NIC's line rate times
window length $\winlen$, but smaller values based on domain knowledge are
often possible.

\textbf{Step 2: DP query and sensitivity.}
In {\sys}, a DP query measures the buffering queue length $\qlen$ with DP at
intervals $\dpintvl$. This noisy measurement determines the number of bytes
that must be transmitted in the interval.
To make the measurement of $\qlen$ differentially private, we need to bound the
sensitivity $\qsens$ of the queue length variable $\qlen$.
This sensitivity $\qsens$ is the maximum difference in $\qlen$ that can be
caused by changing one application stream to neighboring one. Formally, consider
two alternative neighboring streams $\istream$ and $\istream'$ passing through
the queue.
Suppose that when transmitting $\istream$ (similarly $\istream'$), the
queue length at time $k$ is denoted by $\qlent{k}$ (respectively $\qlent{k}'$).
Then, assuming w.l.o.g. that $k \geq k'$:
\setlength{\abovedisplayskip}{0pt}
\begin{equation}
    \qsens = \max_{k = 0}^{\streamduration}~\max_{\istream,
        \istream'} | \qlent{k} - \qlent{k}' |
    \label{eqn:ssens}
\end{equation}

Bounding $\qsens$ is still challenging though, as our neighboring definition
only bounds the difference in traffic between two streams over any window of
length upto $\winlen$.
Because of this, when $\streamduration >> \winlen$, differences between what
$\istream$ and $\istream'$ would enqueue in $\buffQ$ can accumulate over time,
and the difference betwen $\qlent{t}$ and $\qlent{t}'$ can grow unbounded over
time.

To bound $\qsens$, {\sys} relies on the key assumption that
the tunnel can always transmit all incoming data from
application streams within any $\winlen$-sized time window. That is,
\begin{assumption}\label{assumption:window}
All bytes enqueued prior to or at time $t$ are transmitted by time $t+W$.
\end{assumption}
To enforce this assumption, {\sys} implements a time to live in the buffering
queue, flushing all bytes older than $\winlen$ from $\buffQ$ (see
\S\ref{sec:design} for more details).
%
Intuitively, this assumption caps the accumulation of traffic differences in
$\buffQ$ to the maximum difference over $\winlen$, \ie $\ssens$.
Since the size of $\buffQ$ cannot differ by more than $\ssens$ when changing a
stream by a neighboring one, the difference between DP query results of $\buffQ$
under two neighboring streams---which is the sensitivity of the DP query,
$\qsens$---is upper-bounded by $\ssens$. Formally:

\begin{proposition}\label{prop:sensitivity}
    {$\sys$} enforces $\qsens \leq \ssens$.
\end{proposition}

\begin{proofsketch}
Consider any two streams $\istream$ and $\istream'$, as in \Cref{eqn:ssens},
and any measurements time $k$.
The proof proceeds in two steps. First, under \Cref{assumption:window},
streams can accumulate queued traffic for at most {$\winlen$}, so two
different streams can create a difference $|\qlent{k} - \qlent{k}'|$ of at
most $\ssens$.
Second, dequeuing can only make two different queues closer: Consider
query time $k$, with queue lengths $\qlent{k} > \qlent{k}'$ (the
opposite case is symmetric).
For a DP noise draw $z$, we have $\qlendpt{k} > \qlendpt{k}'$. Since shaping
sends at least as much data under $\qlendpt{k}$ as under $\qlendpt{k}'$,
but no more than $\qlendpt{k} - \qlendpt{k}'$, after dequeuing we have
$|\qlent{k+1}' - \qlent{k+1}| \leq |\qlent{k}' - \qlent{k}|$.
In summary, the queue difference under two different streams
$\qsens$ can grow to at most $\ssens$ due to data queuing, and dequeuing only
decreases that difference, and hence $\qsens \leq \ssens$.  The complete
proof is in \S\ref{appendix:dp}.
\end{proofsketch}

\textbf{Step 3: Adding Noise.}
With sensitivity bounded at $\ssens$, we can now query $\qlen$ with DP using
an additive noise mechanism, which entails sampling noise $z$ from a DP
distribution and computing the DP buffer queue length as $\qlendpt{k} \triangleq
\qlent{k} + z$.
Specifically, {\sys} uses the Gaussian mechanism, in which
the noise $z$ is sampled from a centered normal distribution
$z \sim \mathcal{N}\big(\mu,~\sigma^2\big)$, where the variance is parameterized
by $\epsilon_\dpintvl$, $\delta_\dpintvl$, and~$\ssens$:
$\sigma^2 =
(2\ssens^2)/(\varepsilon_\dpintvl^2)\ln(1.25/\delta_\dpintvl)$.
Parameters $\epsilon_\dpintvl$, $\delta_\dpintvl$ determine the amount of
noise added to each DP query~result.

\subsection{Privacy Analysis}
\label{subsec:dp-queue-measurements}

The previous section defines $(\varepsilon_{\dpintvl}, \delta_{\dpintvl})$-DP
guarantees for the traffic transmitted in an individual shaping interval. We now
discuss (i) the guarantees for longer
application streams, (ii) the guarantees on a packet-level sequence derived from
a shaped buffer sequence, and (iii) the privacy implications for streams that
fall outside of the neighboring definition.

\textbf{Guarantees for streams.}
Recall from the previous section that shaping happens at intervals
$\dpintvl$; in each interval we perform a DP query on the buffering queue length
$\qlen$ and create a shaped buffer of length $\qlendp$, which is subsequently
queued for transmitting over the network at the end of the interval.
%
Enqueueing data into the shaped buffer at the end of each shaping interval
creates a sequence of states for the shaped buffer $\{(\qlendp_1, T),
(\qlendp_2, 2T), (\qlendp_3, 3T), \dots \}$.
Although this sequence is not technically observable by an adversary, this is
where we prove our DP guarantees using DP composition over queries
performed during the stream transmission.

\begin{proposition}\label{prop:dp}
For any stream $\istream$ of duration $\streamduration^S \leq
\streamdurationcfg$, {$\sys$}
enforces $(\varepsilon, \delta)$-DP for the sequence $\{(\qlendp_1, T),
(\qlendp_2, 2T), (\qlendp_3, 3T), \dots \}$, with $\varepsilon, \delta
\triangleq \textrm{DP\_compose}(\varepsilon_T, \delta_T, \lceil
\frac{\streamdurationcfg}{\dpintvl} \rceil)$.
\end{proposition}

\begin{proof}
Consider two neighboring streams $\istream$ and $\istream'$. By design, the
times at which shaped buffers are queued are independent of application data, so
$\{(\qlendp'_1, T'), (\qlendp'_2, 2T'), (\qlendp'_3, 3T'), \dots \} =
\{(\qlendp_1, T), (\qlendp_2, 2T), (\qlendp_3, 3T), \dots \}$, and we can
restrict our considerations to the sequences $\{\qlendp_1, \qlendp_2, \qlendp_3,
\dots \}$ and $\{\qlendp'_1, \qlendp'_2, \qlendp'_3, \dots \}$.
By Prop. \ref{prop:sensitivity}, the sensitivity of each measurement is at most
 $\ssens$.
By the Gaussian DP mechanism, the measured queue size $\qlendp$ in each
interval is $(\varepsilon_{T}, \delta_{T})$-DP.
Using DP composition over the $\lceil \frac{\streamdurationcfg}{\dpintvl}
\rceil$ DP queries made during any duration $\streamdurationcfg$
 yields the ($\varepsilon, \delta$)-DP guarantee.
\end{proof}
We use R\'enyi-DP composition on the Gaussian mechanism for
$\textrm{DP\_compose()}$. {\sys} provides a DP guarantee for streams of any
length $t$, with the guarantee degrading gracefully as DP composition (of
order $\sqrt{\streamdurationcfg}$ as the length grows).


\textbf{Guarantees for packet sequences.}
Data from the shaped buffer is sent over the network, and transmitted as a
packet sequence denoted by $\ostream = \{{P^O_1}, {P^O_2}, {P^O_3} \dots \}$.
These packets are a post-processing of the DP-shaped buffer. As long as the
packets are generated independently of any secret data, they preserve the
DP guarantees of shaping due to the post-processing property of DP.
This yields the following result, directly implied by \Cref{prop:dp}
and DP post-processing:
\begin{corollary}
\label{cor:dp}
For any stream $\istream$ of duration $\streamduration^S \leq
\streamdurationcfg$, {$\sys$}
enforces $(\varepsilon, \delta)$-DP for its output packet sequence
\ostream,
with $\varepsilon, \delta \triangleq \textrm{DP\_compose}(\varepsilon_T,
\delta_T, \lceil \frac{\streamdurationcfg}{\dpintvl} \rceil)$.
\end{corollary}

\textbf{Privacy for non-neighboring streams.}
Finally, if the distance between two streams is larger than $\ssens$, \eg say
$k\ssens$ for some factor $k$, {\sys} still provides a (degraded) DP guarantee
through group privacy \cite[Theorem 2.2]{dwork2014algorithmic} applied to the
Gaussian mechanism.
Namely, a DP query in each shaping interval $\dpintvl$ for the non-neighboring
stream provides $(k\epsilon_\dpintvl, \delta_\dpintvl)$-DP, and the guarantees
can be extended for the stream duration $\streamduration$ by applying DP
composition to this new value.

\textbf{Interpretation of the guarantees.}
\update{
\sys's shaping algorithm provides a ($\varepsilon_{\dpintvl},
\delta_{\dpintvl}$)-DP guarantee on the volume of application
traffic enqueued (in the buffering queue) in a fixed-length interval
(and, by post-processing, the volume of traffic observable on the network).
}
%
Under perfect timing for the DP shaping interval $\dpintvl$, the DP guarantee
ensures that two neighboring streams, \ie their contents, are indistinguishable
with the probability as defined in \Cref{eq:dp}. Smaller
$\varepsilon_{\dpintvl}$ and $\delta_{\dpintvl}$ implies higher noise in
shaping, which increases the uncertainty about the original stream in the shaped
traffic.

Secondly, to enforce the DP guarantees, the DP noise of the Gaussian
mechanism does not depend on the total number of flows. Intuitively, the DP
guarantee is shared among all the streams multiplexed through the buffering
queue simultaneously for a fixed amount of noise, and the overhead (\ie noise
added) gets amortized among the streams.

%
%



\begin{figure}[t]
    \centering
    \includegraphics[width=\columnwidth]{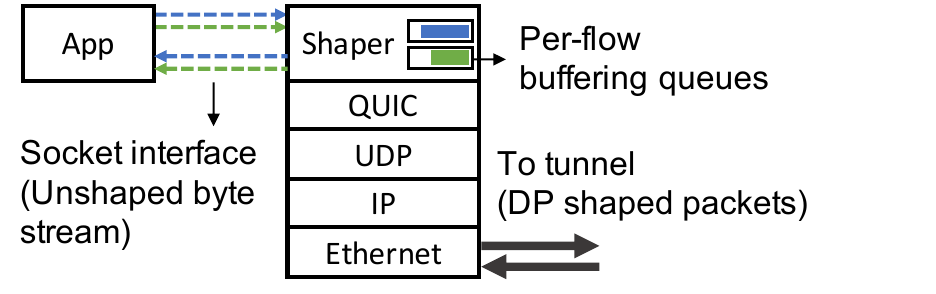}
    \caption{Overview of tunnel design (one endpoint)
    }
    \vspace{-0.4cm}
    \label{fig:minesvpn-overview}
\end{figure}

\textbf{Privacy vs Performance.}
\sys's shaping mechanism introduces several parameters which impact privacy and
performance, specifically latency and bandwidth overheads. These parameters
include $\varepsilon_\dpintvl$, $\delta_\dpintvl$, $\ssens$, $\winlen$, and
$\dpintvl$.
%
Larger $\ssens$ requires lower $\varepsilon_{\dpintvl}$ for stronger privacy,
which implies noisier measurements.
Noisy measurements imply more overhead---when the noise is positive, dummy bytes
need to be sent, incurring higher bandwidth overhead; when the noise is
negative, fewer bytes are sent and the unsent bytes accumulated in the buffering
queue incur a latency overhead.
This is the privacy-overhead trade-off we expect from DP.

Additionally, the parameters $\dpintvl$ and $\winlen$ have a more subtle impact
on the privacy-overhead trade-off.
Traffic is shaped in intervals of $\dpintvl$; thus, $\dpintvl$ impacts the
latency and burstiness of the traffic.
A smaller $\dpintvl$ provides lower transmission latency and smaller bursts per
interval. However, it also requires more DP queries and, thus, incurs higher
privacy loss when transmitting the complete stream of a given length
$\streamduration$.
In practice, one would set $\dpintvl$ to the maximum value that can minimize
privacy loss while providing tolerable latency.

A large $\winlen$ for a fixed $\ssens$ implies that
neighboring streams can differ by at most $\ssens$ over a longer window size
$\winlen$, which weakens the neighboring definition and, hence, the privacy
guarantees.
While smaller $\winlen$ is desirable, the lower bound is $\dpintvl$.
Recall that, to bound $\qsens$ (\Cref{assumption:window}),
{\sys} must drop any bytes left in the buffering queue for longer than $\winlen$.
Since data leaves the buffering queue in each shaping interval after a
DP query, setting $\winlen = \dpintvl$ would lead to immediate dropping of the
bytes from the buffering queue that aren't transmitted in response to the
DP query. This would happen \update{in each interval where the DP query samples
a negative noise value}.
These data drops would degrade {\sys}'s performance in terms of both
throughput and latency. Hence, we need $\winlen >> \dpintvl$ to ensure that data
has time to leave the buffering queue before it is too old.
Typically, one would set $\winlen$ based on application domain knowledge, such
as the maximum size of a web request or the fact that videos often consist of a
sequence of segments requested at 5s intervals.

We analyze the impact of different choices for these parameters on
the privacy guarantees and overheads in \S\ref{sec:eval}.

\section{Traffic Shaping Tunnel}
\label{sec:design}

%
A tunnel must address three requirements.
First, it must satisfy DP guarantees. For this, the tunnel~must complete DP
queries and prepare shaped packets within each interval, and it
must be able to transmit all payload bytes generated from an application within
a finite window length (as defined in the DP strategy).
Secondly, the payload and dummy bytes in the shaped packets must be
indistinguishable to an adversary. For this, the payload and dummy bytes must be
transmitted through a shared transport layer so that they are identically
acknowledged by the receiver and subject to congestion control and
loss recovery mechanisms.
Finally, the tunnel must provide similar levels of reliability,
congestion control, and loss recovery as expected by the application.

\Cref{fig:minesvpn-overview} shows the design of one endpoint of {\sys}’s
traffic shaping tunnel. A similar endpoint is deployed on the other end of
the tunnel.
%
{The shape of the traffic in the tunnel can be configured independently
in each direction. The privacy loss in bidirectional streams is the DP
composition of the privacy loss in each direction.}

A tunnel endpoint consists of a shaping layer (Shaper) on top of QUIC, which
in turn runs on top of a standard UDP stack.
The tunnel endpoints establish a
bidirectional QUIC connection and generate DP-sized transmit buffers in fixed
intervals, which carry payload bytes from one or more application flows. In the
absence of application payload, a tunnel endpoint transmits dummy bytes, which
are discarded at the other endpoint. QUIC encrypts all outbound packets.


{\sys} adopts a transport-layer proxy architecture: each application
terminates a connection with its local tunnel endpoint. The application byte
stream is sent to the remote application over three piecewise connections: (i)
between the application and its local tunnel endpoint, (ii) between the
tunnel endpoints, and (iii) between the remote tunnel endpoint and the remote
application.
This ensures only one active congestion control and reliable delivery mechanism
in the tunnel and that all bytes are subject to identical mechanisms\footnote{
We discard tunneling TCP through TCP as it causes TCP
meltdown~\cite{honda2005tcpovertcp, tcp-meltdown}, or TCP through UDP as it is
unsafe.
(TCP between the application hosts would retransmit lost payload bytes
only, not any dummy bytes injected between the tunnel endpoints, making the
dummy bytes observable.)
}.


The application and the tunnel endpoint shown in \Cref{fig:minesvpn-overview}
could either be colocated on the same host or located on separate hosts.
In each case, the traffic between the application and the tunnel endpoint is
assumed to be unobservable to an adversary.
Our design (\S\ref{subsec:design-overview}) does not distinguish between the two
configurations.
Our implementation (\S\ref{sec:implementation}) assumes that the tunnel endpoint
is located on a separate middlebox. We discuss security in
\S\ref{subsec:impl-security} and alternate deployments in
\S\ref{subsec:design-discussion}.

\subsection{Tunnel Design and Operations}
\label{subsec:design-overview}

\textbf{Tunnel setup and teardown.}
\update{Before applications can communicate with each other, a {\sys} tunnel
must be set up between their local tunnel endpoints.
The initiator application sends a configuration message to its local tunnel
endpoint with the source and destination IP addresses and ports, a reliability
flag, and a privacy descriptor.}
The reliability flag indicates if the tunnel should provide reliable delivery
semantics or not. The privacy descriptor indicates the DP parameters to be used
for shaping the tunnel traffic.

Upon receiving a configuration message, the Shaper establishes a QUIC
connection with the remote tunnel endpoint and configures the reliability
semantics and privacy parameters
for each direction.
It also initializes \todo{three types of bidirectional~streams in the tunnel:
control, dummy, and data streams}.
One {\em control stream} is used to transmit
messages related to the establishment and termination of a connection
between the application endpoints. A {\em dummy stream} transmits padding
in QUIC packets in the form of STREAM frames\footnote{We do not use QUIC's
PADDING frames as they do not elicit acknowledgements and hence are
distinguishable from STREAM frames \cite{rfc9000}.}.
\todo{The tunnel pre-configures a finite number of data streams, which carry
payload bytes from one or more application flows.}

{When the tunnel is inactive for a period of time, one of the tunnel
endpoints initiates a termination sequence and closes all open QUIC streams and
the tunnel connection.}

\textbf{Connection establishment and termination.}
Once a tunnel is ready, applications can establish and terminate connections
with each other, which is mediated by the tunnel.
When the initiator application runs a connection establishment handshake with
its local tunnel endpoint, the Shaper maps the application flow to
a per-flow buffering queue and one of the inactive QUIC data streams in the
tunnel, and notifies the remote tunnel endpoint.
The remote tunnel endpoint establishes a connection with the receiver
application and maps the receiver application's flow with the data stream.
The connection termination handshake is handled similarly by the tunnel
endpoints. The messages for connection establishment and termination are
transmitted over the control stream in the tunnel and shaped according to the
tunnel's parameters.

\textbf{Outbound traffic shaping.}
The Shaper accumulates the outbound bytes of an application flow in a
buffering queue before it transmits them in packets whose sizes and timing
follow a distribution that guarantees DP.
Within a tunnel, the Shaper transmits bytes from all active flows into a
differentially-private packet sequence.
At periodic intervals, {called DP shaping intervals,} it performs a
DP query on the per-flow queues to determine the
number of bytes $\qlendp$ to be~transmitted according to the tunnel's DP
parameters. It prepares a {\em shaped buffer} consisting of $\payload$ payload
bytes and $\dummy$ dummy bytes, where $\payload$ is the minimum of $\qlendp$ and
the application bytes available in the buffering~queues, and $\dummy = \qlendp -
\payload$, which may lie between 0 and $\qlendp$. The Shaper then passes
the buffer with the position and length of the padding to QUIC.

QUIC transforms the shaped buffer into one or more STREAM frames based on the
congestion window, the flow window of the receiver endpoint, and the MTU
(maximum transmission unit). It
places the padding bytes into a dummy STREAM frame. QUIC packages the frames
into packets, whose length is at most MTU
minus the length of the headers and whose payload is
encrypted. QUIC forwards the packets to the UDP layer, which subsequently
transmits the prepared packets as quickly as it can, given the line rate of the
NIC.

To enforce \Cref{assumption:window}, the Shaper tracks the expiry time of
each byte enqueued in $\buffQ$ based on the arrival time and the neighboring
window length $\winlen$ configuration. The Shaper drops the untransmitted bytes
in the queue upon their expiry.

%


\textbf{Inbound traffic processing.}
A tunnel endpoint receives shaped packets from the tunnel and applies inverse
processing on each packet.
QUIC receives the packet and
sends an ACK to the sender. Subsequently, it decrypts the packet, discards the
dummy frame, and forwards the payload bytes from the remaining STREAM frames to
the application.

\subsection{Middlebox Implementation}
\label{sec:implementation}

\begin{figure}[t]
    \centering
    \includegraphics[width=\columnwidth]{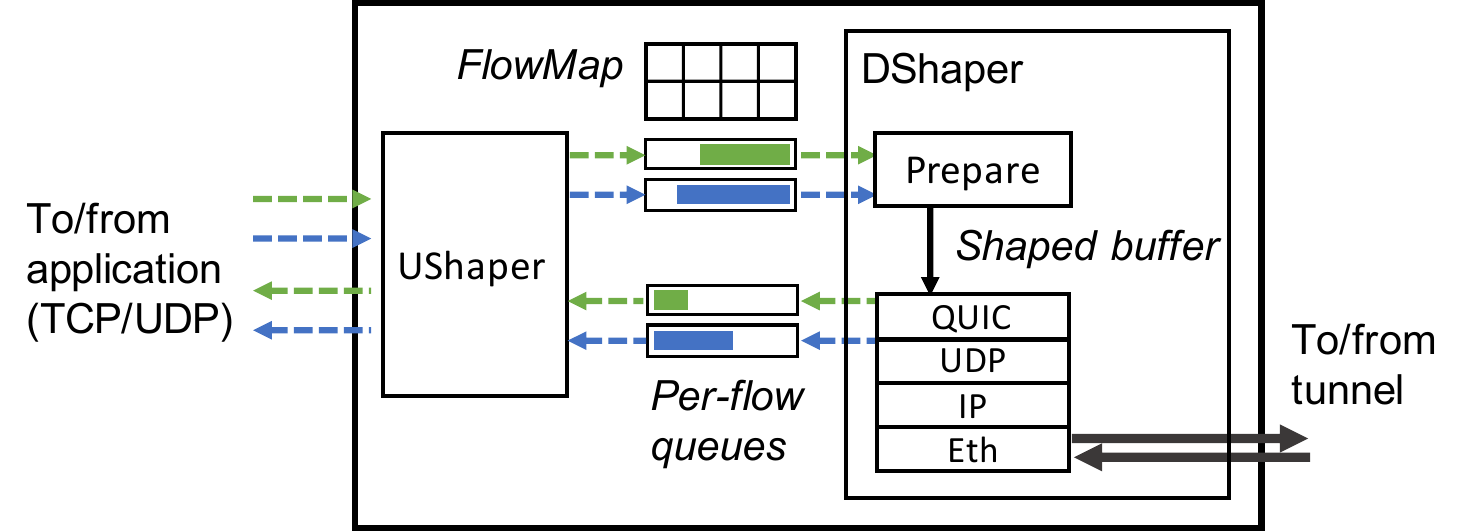}
    \caption{{\sys} middlebox design}
    \vspace{-0.4cm}
    \label{fig:minesvpn-impl}
\end{figure}

We present a middlebox-based {\sys} implementation.
%
%
For ease of implementation, our prototype requires applications to explicitly
connect to the middleboxes.
In principle, {\sys} can transparently proxy application
connections.

In our implementation (see \Cref{fig:minesvpn-impl}), a middlebox consists of
two userspace processes. The
{\ushaper} mediates {\em unshaped} traffic between the applications and the
middleboxes. The {\dshaper} handles {\em DP shaped} traffic within the tunnel.

\textbf{UShaper.}
The {\ushaper} implements a transport server (or client) for interfacing with
each local client (or server, respectively) application.
For managing multiple flows, it shares a {\flowmap} table with the {\dshaper},
which consists of an entry for each end-to-end flow. Each entry maps the
piecewise connections with
a pair of transmit and receive queues to carry the local application's byte
stream, and shaping configurations (\eg privacy descriptor) provided by an
application at the time of flow registration.

The {\ushaper} receives the outbound traffic from a sender application
and enqueues the byte stream into a per-flow transmit queue shared with the
{\dshaper}.
It also dequeues bytes from a per-flow receive queue, repackages them
into transport packets and sends them to the receiver application.

\begin{figure}[t]
    \centering
    \includegraphics[width=\columnwidth]{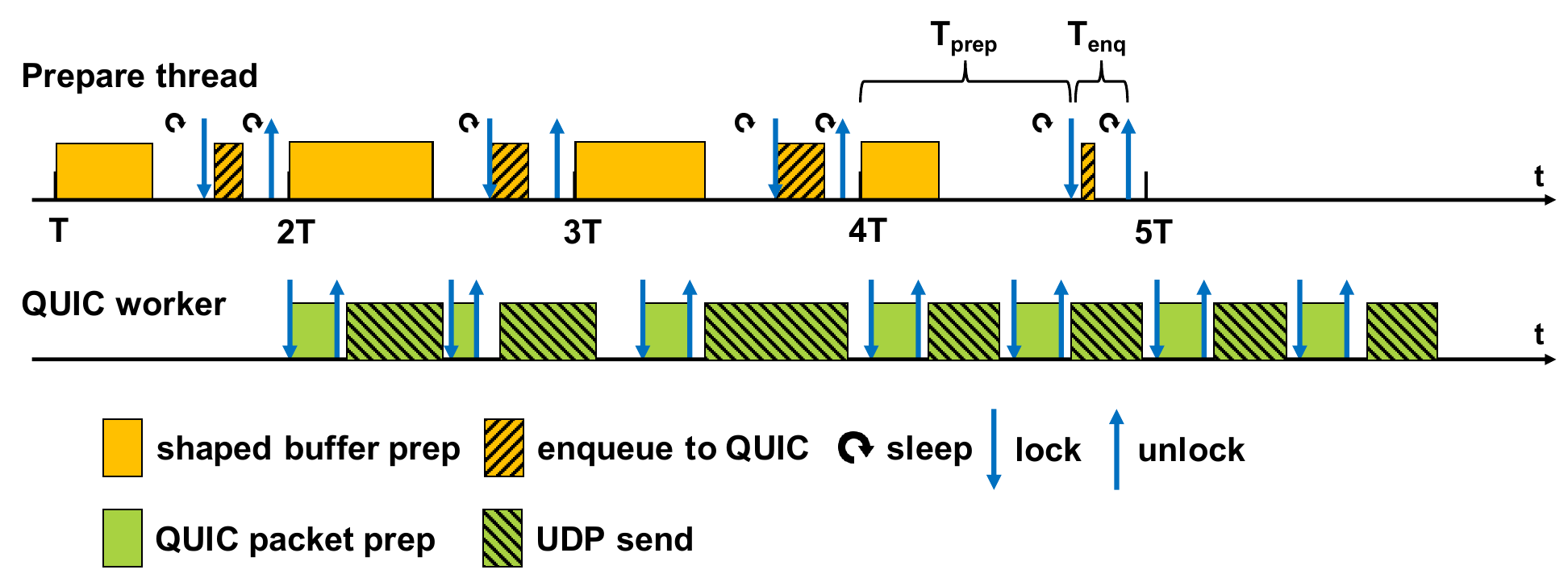}
    \caption{{\dshaper} schedule}
    \vspace{-0.4cm}
    \label{fig:middlebox-schedule}
\end{figure}

\textbf{DShaper.}
The {\dshaper} consists of a {\prepare} thread and a QUIC
worker thread.
The {\prepare} thread instantiates a QUIC client/server to establish
a tunnel with the remote middlebox and implements the DP shaping logic.
On the transmit side, {\prepare} {prepares shaped buffers based on DP
queries on the transmit queues}
and then submits shaped buffers to the QUIC worker for transmission.
On the receive size, the QUIC worker transmits ACK frames to the sender
and then decrypts the QUIC packets, extracts the STREAM frames, and copies
bytes (including dummy bytes) from each frame into the appropriate per-flow
receive queue.

\textbf{Ensuring secret-independent shaping.}
To enforce DP guarantees, {\dshaper} should {\em transmit} exactly $\qlendp$
bytes in each DP shaping interval $\dpintvl$.
%
%
This would require ensuring:
{\bf P1.} {\prepare} computes $\qlendp$, allocates and prepares a buffer
of length
$\qlendp$, and passes the buffer to the QUIC worker within $\dpintvl$,
{\bf P2.} the QUIC worker prepares encrypted packets from the buffer and sends
them to UDP, such that the total payload size of the QUIC packets prepared in
$\dpintvl$ is $\qlendp$, and
{\bf P3.} the~UDP stack transmits packets totaling to $\qlendp$ payload bytes
to the NIC in $\dpintvl$.
Enforcing all these properties would require a constant-time implementation for
each step, which is non-trivial, or a strict time-triggered schedule
for each step, which would significantly reduce link utilization and increase
packet latencies.

\todo{Thanks to DP post processing, however, it suffices to ensure the property
P1,
and {\bf P4.} that the QUIC worker transforms the shaped buffer into network
packets independently of the
application data.
No other constraint on the sizes and timing of network packets is required to preserve DP.
}

Satisfying P1 involves one challenge. Although the application is physically
isolated from the middlebox, its flow control behavior could be
secret-dependent and could affect the middlebox's execution.
For instance, the presence or absence of payload traffic from an application
can affect the time {\dshaper} requires to prepare the shaped buffers.

Thus, {\dshaper} satisfies P1 as follows (see \Cref{fig:middlebox-schedule} for
reference).
First, {\prepare} {guarantees that $\qlendp$ is computed with a DP query
in each interval}.
Secondly, {\prepare} guarantees that a shaped buffer of length $\qlendp$ is {\em
prepared} within a fixed time $\dpintvl_{prep}$ within each interval.
Thirdly, {\prepare} locks the shaped buffer for a fixed time,
$\dpintvl_{enq}$, during which it enqueues the buffer for a QUIC
worker.
This ensures that~the buffer is completely enqueued before QUIC starts
transmitting it and that QUIC receives the buffer only at fixed delays.

We empirically profile the time taken by {\prepare} for preparing and enqueueing
shaped buffers for various DP lengths. We set $\dpintvl_{prep}$ and
$\dpintvl_{enq}$ to maximum values determined from profiling, and $\dpintvl$ to
the sum of these maximum values, \ie $\dpintvl_{max}$.
If {\prepare} takes time less than
$\dpintvl_{prep}$ (or $\dpintvl_{enq}$, respectively) to prepare (or enqueue) a
shaped buffer, it sleeps until the end of the interval before moving to the next
phase.

To satisfy P4,
{\prepare} and QUIC worker threads run on separate cores
sharing only the shaped buffers.
{\ushaper} runs on yet a different core and shares the {\flowmap} and the
per-flow transmit queues containing unshaped traffic only with {\prepare}. It
shares the per-flow receive queues with the QUIC worker, but they contain only
shaped frames from the QUIC worker.
\update{Finally, we assume that QUIC encrypts and decrypts shaped
buffers in constant-time.}
With this, the execution of the QUIC worker becomes independent from {\prepare}
and secret-independent overall.
%
Consequently, the QUIC worker and the UDP stack can packetize the shaped buffers
and transmit the packets at link speed, and any variance in packet
transmit times constitute post-processing noise.

\textbf{Limitations.}
Our prototype has two limitations in enforcing secret-independent timing.
First, our QUIC implementation uses
standard OpenSSL, which may not provide constant-time crypto. However, QUIC can
be modified to adopt a constant-time crypto library~\cite{hacl,libsecp256k1} to
overcome this limitation.
Secondly, it is difficult to find the true maximum values of
$\dpintvl_{prep}$ and $\dpintvl_{enq}$ on general-purpose desktops. If
{\prepare}'s execution exceeds the profiled max values, it violates the
theoretical DP guarantees. However, we note that it is difficult to practically
exploit these violations for inferring traffic secrets.

\subsection{{Security Analysis}}
\label{subsec:impl-security}
{\sys} provides the following security property: an adversary cannot infer
application secrets from observing tunnel traffic. This
property is ensured by a combination of a secure shaping strategy, the tunnel design,
and implementation.

{\bf S1. Secure shaping strategy.} The tunnel transmits traffic in
differentially private-sized bursts
in fixed intervals. Thus, the overall shape is DP. The proof of DP is in
\S\ref{appendix:dp}.

{\bf S2. Secure tunnel design.}
(i) The privacy guarantees of a tunnel are configured before the start of
application transmission and do not change during the tunnel's lifetime.
(ii)~The tunnel mediates control between the end hosts, \eg by
transmitting custom connection establishment and termination messages. These
messages are subject to the same DP shaping as the payload traffic
(\S\ref{subsec:design-overview}).
(iii) The payload and dummy bytes in network packets are indistinguishable
because all payload and dummy bytes are packaged into QUIC packets and
encrypted uniformly. Moreover, QUIC handles acknowledgements, congestion
control, and loss recovery for both payload and dummy bytes uniformly
(\S\ref{subsec:design-overview}).

{\bf S3. Secure middlebox implementation.}
{(i)} The unshaped traffic between an end host and its local middlebox is
not visible to an adversary.
(ii) {\dshaper} follows the tunnel design in transmitting payload and dummy
bytes.
%
%
(iii) The time required for {\prepare} to prepare and enqueue shaped
buffers is masked to secret-independent times. The packetization of~buffers in
QUIC is secret-independent~(\S\ref{sec:implementation}) and
thus retains DP guarantees after post-processing.
Any delays in transmitting the buffers can arise only due to congestion or
packet losses in the tunnel network, which are secret-independent events.

\subsection{Deployment and Maintenance}
\label{subsec:design-discussion}
{\sys}'s tunnel endpoint design and implementation are both very modular and
portable.
The middlebox components are compatible with all application layer protocols
(\eg HTTPS, QUIC-TLS) and network stacks (\eg TCP, UDP, QUIC stacks).
The {\ushaper} could also be implemented as a standard SOCKS5 proxy~\cite{torpt}.

A tunnel endpoint could be integrated with any node along an application's
network path as long as the application traffic is unobservable until egress
from the tunnel.
By integrating with a trusted VPN gateway of an organization,
network administrators could manage ``long-term'' tunnels between the
organization's campuses and support multiple applications without modifying
individual end hosts. For instance, separate tunnels may be configured
per-application according to the organizational needs, or
configurations may be adapted based on coarse-grained changes in the traffic
patterns through the~day.

Alternatively, by integrating with end user devices, \eg with VPN clients,
users could instantiate a new bidirectional tunnel with a service before
each network activity, choose a different configuration for each tunnel
instance, and close the tunnel after completion of the activity. A key
requirement would be to secure the tunnel endpoint's execution from any internal
side channels on the end host, which would be now more prevalent than in the
middlebox setup.

\section{Evaluation}
\label{sec:eval}

Our evaluation answers the following questions.
(i) How well does {\sys} mitigate state-of-the-art network side-channel attacks?
{(ii) What are the overheads associated with varying DP relevant
configuration parameters?}
(iii) What are the packet latency overheads and the peak line rate and
throughput sustained by our {\sys} middlebox?
(iv) What are {\sys}'s overall costs on privacy, bandwidth, client latency, and
server throughput for different classes of applications?
(v) How do {\sys}'s privacy guarantees and performance overheads compare to
prior techniques?

For our experiments, we use four AMD Ryzen 7 7700X desktops each with eight 4.5
GHz CPUs, 32 GB RAM, 1 TB storage, and one Marvell AQC113CS-B1-C 10Gbps NIC. We
simulate client and server applications on two of the desktops and {\sys}'s
middleboxes on the other two desktops.
The middleboxes are connected to each other via an additional Intel X550-T2
10Gbps NIC on each desktop.
The client and server desktops are connected to one of the two
middleboxes each via the Marvell NICs, overall forming a linear topology.

We implemented the {\ushaper} and {\dshaper} processes in
\update{1100} and \update{1800} lines of C++ code, respectively, and deployed
them on Ubuntu OS 22.04.02 (kernel version 5.19).
{\sys} relies on the MSQUIC implementation of QUIC, libmsquic v2.1.8, which
includes OpenSSL for traffic encryption, contributing an additional
\update{180713} LoC to {\sys}'s TCB.

For rapid evaluation, we  built a simulator,
which implements the {\prepare} thread's DP logic.
The simulator transforms  an application's original packet sequence (from
tcpdump) into a sequence of burst sizes within fixed-length
intervals, and outputs a sequence of transmit sizes corresponding to DP
queries.
{We confirmed that the bandwidth overheads from the simulator closely match
the overheads observed on the testbed. Thus, we report privacy and
bandwidth overhead results from the simulator and latency and
throughput results from the testbed, unless specified otherwise.}
%

We use two applications for case studies, a video streaming service and a
medical web service, which we describe below. Both applications are hosted on an
Nginx 1.23.4 web server and the datasets are stored on the host file system.

\textbf{Video service.}
The video streaming service is used to serve \update{100 YouTube videos} in 720p
resolution with \update{5 min to 130.3 min (median 12.6 min)} durations and
\update{2.7 MB to 1.4 GB (median 73.7 MB)} sizes.
We implement a custom video streaming client in Python, which
uses one TCP connection for a single stream and requests individual 5s segments
from the service synchronously.
Unless specified otherwise, we stream the first 5 min of videos in our
experiments. The configuration is in line with that used in prior
work~\cite{schuster2017beautyburst} and reasonable, given the popularity of
short videos~\cite{videostats}.

\textbf{Web service.}
The web service is used to serve a corpus of {96} static HTML pages of a medical
website, whose sizes range from \update{54 KB to 147 KB (median: 90 KB)}.
{As a client, we use modified wrk2 \cite{wrk2} that issues concurrent
asynchronous HTTPS GET requests at a specified rate.}

For all experiments, we use one baseline setup and one of three {\sys}
configurations.
In the baseline setup ({\base}), the client is directly connected to the server.
In the simulator setup ({\nssim}), we generated sequences of shaped burst sizes
using DP shaping.
With {\sys}, the traffic between the client and the server passes through two
middleboxes, each implementing {\ushaper} and {\dshaper}.
We consider two configurations of the middleboxes:
(i) {\nsnoshape}: {\dshaper} does not implement shaping (\ie neither DP noise
sampling nor the fixed length loop interval $\dpintvl_{max}$), allowing us to
measure the system overheads due to the middlebox implementation, and
(ii) {\ns}: {\dshaper} implements the full shaping mechanism.
{By default, each middlebox is configured with 128 pairs of per-flow
transmit and receive queues (unless specified otherwise). We configure the queue
sizes for the max data that can be transmitted at line rate for a given
DP shaping~interval. We configure a max cutoff for the shaped buffer length
based on the max burst length of the application and the number of active
flows. This implies that the number of flows is public.}

In addition, we compare with two other shaping strategies: constant-rate shaping
({\constshape}), which is the most secure shaping strategy, and Pacer
({\pacer}), a SOTA system that shapes traffic on a per-request basis.
For {\constshape}, we configure the peak load corresponding to the largest
object sizes in our applications, which involves transmitting 1.7MB in 5s for
videos and 57KB in 50ms for web.
For Pacer, we pad all web pages to the largest page size, \ie 147KB in our
dataset. For videos, Pacer pads a segment at $i$\textsuperscript{th} index
in a video stream to the largest segment size at that index across all videos in
the dataset.

\begin{figure}[t]
    \centering
    \includegraphics[width=0.85\columnwidth]{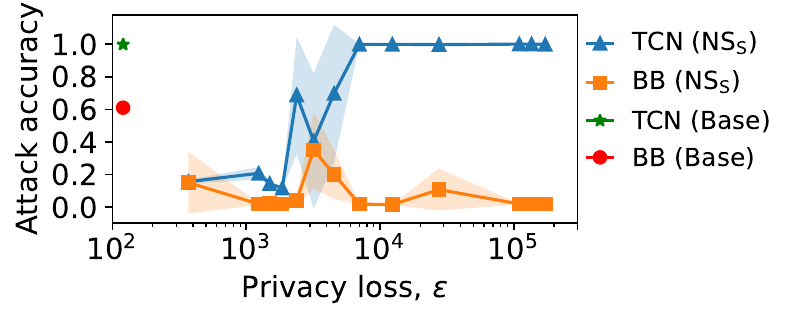}
    \caption{\update{Classifier accuracy on shaped traces.}
    }
    \vspace{-0.4cm}
    \label{fig:empirical-privacy}
\end{figure}

\subsection{{\sys} Defeats Attack Classifiers}
\label{subsec:attack-eval}
We start with an empirical evaluation of the privacy offered by {\sys}'s traffic
shaping. Recall that the traffic shaping depends on several DP parameters: the
window length $\winlen$, the sensitivity $\qsens$, the length of the DP
interval $\dpintvl$, and the privacy loss $\varepsilon$.
We evaluated the classifiers from \S\ref{subsec:attack-bg} on shaped traffic generated using various values of these DP relevant parameters.
We present the classifiers' performance based on~only one set of values for
$\winlen$, $\qsens$, and $\dpintvl$, while varying $\varepsilon$ between
\update{[200, 200000]}.
Our goal is to provide intuition about what values of $\varepsilon$ are
sufficient to thwart a side-channel~attack.

We set (i) $\winlen = 5s$ to align with the 5s video segments that comprise the
videos, (ii) \update{$\qsens = 2.5 MB$}, \update{which covers 99th \%ile}
of the distances in our dataset (\S\ref{sec:eval-extended}), and
(iii) $\dpintvl = 1s$, which
leads to composing the privacy loss over $\varnumupdates = 5$ DP queries on
the buffering queues. The 1s interval provides a reasonable trade-off between
privacy loss, and bandwidth and latency overheads, which we
discuss in the subsequent sections.
%

We used 40 videos of 5 min duration each. We streamed each video 100 times
through our testbed without shaping~and collected the resulting tcpdump
traces. For each value of \update{$\varepsilon$}, we transformed each unshaped
trace into a shaped trace using our simulator to generate a total of 4000
shaped traces.
We also shaped traces for {\constshape} and {\pacer}.


\update{
We train and test BB and TCN on shaped traces as in \S\ref{subsec:attack-bg} and
report the average and standard deviation of the accuracy of each classifier
over three runs.
}
Recall from \S\ref{subsec:attack-bg}, the BB and TCN accuracy on unshaped
traffic ({\base}) is 0.61 and 0.99, respectively.
\update{For {\constshape} and {\pacer}, the accuracy of both BB and TCN
classifiers is 0.025. This is expected since both strategies transform all
unshaped traces into a single shape.
}

\Cref{fig:empirical-privacy} shows the classifiers' performance with {\sys}.
While BB does not perform well for nearly all values of \update{$\varepsilon$},
TCN can be thwarted for \update{$\varepsilon$} upto 1000.
%
{\em While even \update{{$\varepsilon \approx 200$}} is too
large to offer meaningful theoretical privacy guarantees, it is sufficient to
defeat SOTA attacks.}

%

\begin{figure}[t]
    \centering
    \begin{subfigure}{0.49\columnwidth}
        \centering
        \includegraphics[width=\textwidth]{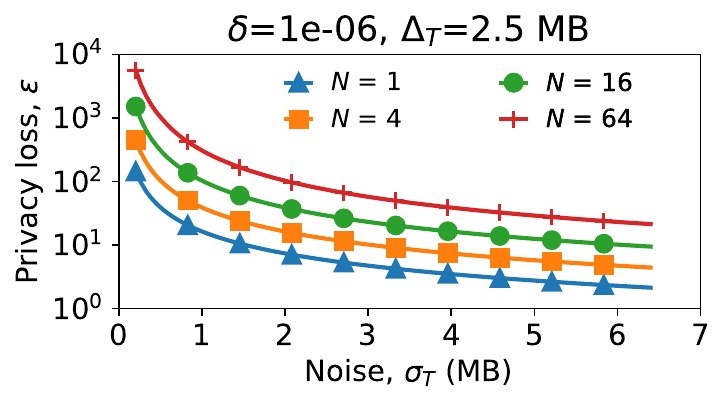}
        \caption{}
        \label{subfig:high-sensitivity-epsilon-sigma}
    \end{subfigure}
    \hfill
    %
    \begin{subfigure}{0.49\columnwidth}
        \centering
        \includegraphics[width=\textwidth]{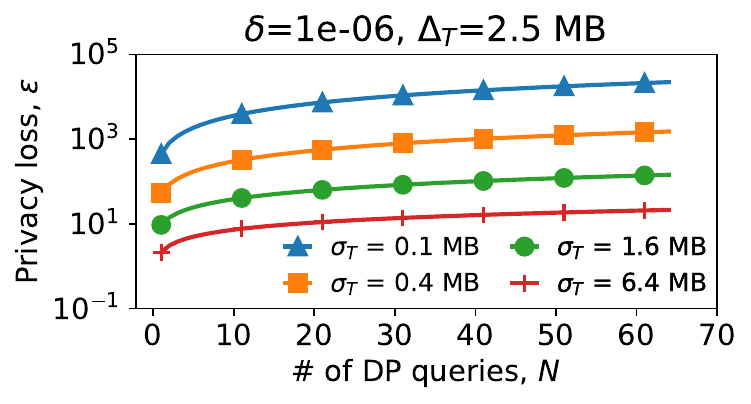}
        \caption{}
        \label{subfig:high-sensitivity-epsilon-queries}
    \end{subfigure}
    \caption{
      \update{Privacy loss vs
      (a) noise and (b) \# of DP queries.}
    }
    \vspace{-0.4cm}
    \label{fig:privacy-microbenchmarks}
\end{figure}

\subsection{Impact of Privacy Parameters}
\label{subsec:privacy-microbenchmarks}

We now  evaluate how $\varepsilon$ varies with $\qsens$,
$\sigma_\dpintvl$, and $\varnumupdates$.
Due to space constraints, we present plots for a fixed value of $\qsens
= 2.5 MB$ and defer other plots to \S\ref{sec:eval-extended}.
All analyses use $\delta = 10^{-6}$.

\Cref{subfig:high-sensitivity-epsilon-sigma} shows the tradeoff between
$\varepsilon$ and $\sigma_\dpintvl$ over four different values of
$\varnumupdates$.
Intuitively, a larger
$\sigma_\dpintvl$ implies higher bandwidth overhead due to DP shaping.
To retain a total privacy loss $\varepsilon = 1$ with at most 4 DP
queries, we need to add noise with $\sigma_\dpintvl = 18 MB$ for each DP
query. In contrast, $\varepsilon = 200$ with 4 DP
queries (approx. configuration that defeats the classifiers
in \S\ref{subsec:attack-eval}) only requires $\sigma_\dpintvl < 0.3 MB$.
%
\Cref{subfig:high-sensitivity-epsilon-queries} shows that the
total privacy loss escalates with an increase in the number of DP queries.
While fewer queries within a window (thus larger decision intervals) help to
lower the total privacy loss, the tradeoff is the higher latency overhead. We
discuss this tradeoff, as well as reduction in bandwidth overheads with
concurrent flows in \S\ref{subsec:eval-video} and \S\ref{subsec:eval-web}.

Using these plots, an application can choose suitable values of $\winlen$ and
\update{$\qsens$} to determine the tradeoff between \update{$\varepsilon$} and
$\sigma_{\dpintvl}$.
For our web application serving static HTML, we recommend $\winlen =
1s$, since web page downloads in our AWS setup (\S\ref{subsec:attack-bg})
finished within 1s, and \update{$\qsens = 60 KB$}, which covers all distances.
\S\ref{subsec:attack-eval} explained the choices for our video application.
\update{
Using \update{$\varepsilon$} and $\dpintvl$, we can further determine the
aggregate privacy loss over longer traffic streams using R\'enyi-DP composition.
For instance, with \update{$\varepsilon = 1$}, $\winlen = 5s$, and $\dpintvl =
1s$, the total privacy loss for a 5 min video, which generates 300 DP
queries at 1s intervals, is 8.92; the total loss for a 1 hr~video is 38.8.
}
We emphasize that {\em \update{$\qsens$ and $\varepsilon$} should be selected
using plots like \Cref{fig:privacy-microbenchmarks},
independently of the application's dataset.}

\subsection{Performance Microbenchmarks}
\label{subsec:perf-microbenchmarks}
We now turn our attention to experiments to determine the overheads on
per-packet latencies and the peak line rate and throughput sustainable by a
{\sys} middlebox.

\begin{figure}[t]
    \centering
    \begin{subfigure}{0.49\columnwidth}
    \centering
    \includegraphics[width=\textwidth]{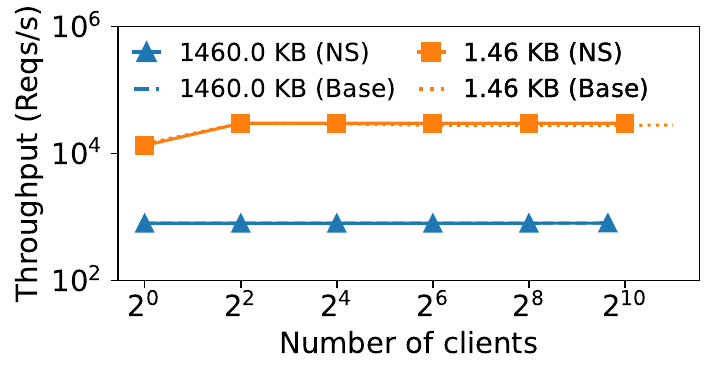}
    \caption{\# clients vs throughput}
    \label{subfig:clients-vs-xput}
    \end{subfigure}
    \hfill
    \begin{subfigure}{0.49\columnwidth}
    \centering
    \includegraphics[width=\textwidth]{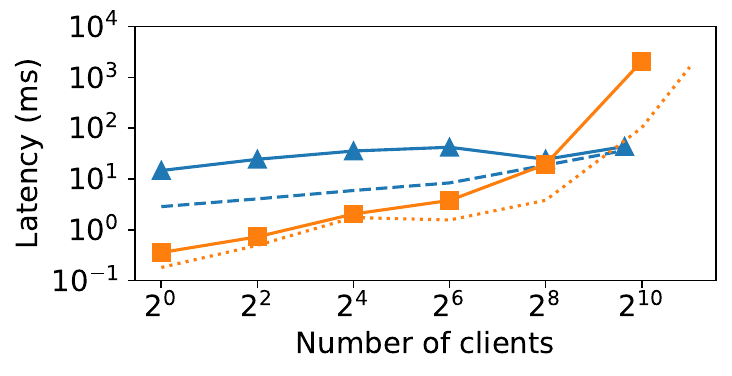}
    \caption{\# clients vs latency}
    \label{subfig:clients-vs-latency}
    \end{subfigure}
    \caption{Throughput and latency overhead due to middleboxes without shaping.
    }
    \vspace{-0.4cm}
    \label{fig:xput-latency-microbench}
\end{figure}

\textbf{Middlebox throughput.}
%
We measure the peak throughput (requests/s) attained by a server application,
the response latency experienced by the clients, and the impact on this
throughput and latency due to the middleboxes.
We restrict Nginx to one worker thread and one core on the server desktop.
We evaluate using two object sizes: 1.4 KB (one MTU) and 1.4 MB.
We identify the peak request rate that can be handled by a single wrk2 client,
then increase the clients until we find the peak throughput the server can
provide.
We then vary the number of concurrent clients while generating the peak request
load sustainable by the server to find the maximum number of concurrent clients
that the server can handle and to measure the impact on the response latency.

We run experiments for \update{3 min} and discard the measurements from the
first minute to eliminate startup effects.
\Cref{subfig:clients-vs-xput} shows the average of the peak throughput
observed across \update{3 runs}.
The standard deviation is \update{below 1\%} in all cases. For 1.4 KB and 1.4 MB
objects, the {\base} server achieves a peak throughput of \update{30K req/s}
\update(64 clients) and \update{800 req/s} (800 clients), respectively.
{\nsnoshape} matches the peak throughput and the max concurrent clients
sustained by {\base}.

\textbf{Latency.}
\Cref{subfig:clients-vs-latency} shows the average and standard deviation of
the response latencies over a \update{2 min} run.
The ping latency between each pair of directly connected desktops is
\update{0.56 $\pm$ 0.18 ms.}
%
This overhead comes from the fact that
each packet traverses four additional network stacks (across two middleboxes) in
each direction. This also involves data copy operations between the kernel and
user space. The data copy overhead is proportional to the object size; thus
{\nsnoshape}'s latency overhead increases with the larger response sizes.

The kernel data copy overheads are not fundamental to {\sys}'s design.
By using kernel bypass techniques or tools like DPDK \cite{dpdk}, {\sys} can
reduce the latency overhead.

\textbf{Shaping interval, preparation, and enqueue times.}
We further profile the middlebox execution to measure the max latencies of
the two components in the {\prepare} loop (\S\ref{sec:implementation}): the
preparation of the shaped buffer and queuing of the buffer to QUIC worker.
These measures determine the maximum
durations for preparing and
enqueueing shaped buffers ($\dpintvl_{prep}$ and $\dpintvl_{enq}$,
respectively), and the minimum value for the shaping interval
$\dpintvl$. We profile the delays with the middlebox configured
with \update{128 queues}, thus supporting a maximum of 128 concurrent clients.
One can profile the delays for a different
number of queues and concurrent clients.

Based on our measurements, we set \update{$\dpintvl_{prep} = 6 ms$ and
$\dpintvl_{enq} = 1 ms$}. The smallest
value for $\dpintvl$ that we can configure is \update{10ms}.

\textbf{Throughput and latency with shaping enabled.}
We now re-run the microbenchmarks with {\ns} configuration. We use three
different configurations for $\dpintvl$: \update{10ms, 50ms, and 100ms}.
We use 128 concurrent clients. The middlebox can sustain the peak throughput of
\update{30K req/s with 1.4KB objects} and \update{700 req/s with 1.4MB} objects
for each configuration of $\dpintvl$. For 1.4KB objects, the average and standard
deviation of the response latency with the three configurations are as follows:
\update{(i) $\dpintvl = 10ms$: 30.47 $\pm$ 3.89 ms}, \update{(ii) $\dpintvl =
50ms$: 51.39 $\pm$ 14.64 ms, (iii) $\dpintvl = 100ms$: 77.49 $\pm$ 28.96 ms}.
For 1.4MB objects, the latencies are as follows:
\update{(i) $\dpintvl = 10ms$: 41.31 $\pm$ 10.84 ms}, \update{(ii) $\dpintvl =
50ms$: 76.96 $\pm$ 21.12 ms}, \update{(iii) $\dpintvl = 100ms$: 127.48 $\pm$
45.69 ms}.
\update{The latency is dominated by the $\dpintvl$ configuration.}
The high variance in the latency is due to shaping.
If a request arrives just after the decision loop has prepared a buffer in the
current iteration, the request will be delayed by at least one iteration of the
loop.
Moreover, a negative sampling of DP noise may lead to a smaller shaped buffer
than the available payload bytes in the buffering queues, thus delaying the
requests by one or more intervals. This effect is particularly enhanced in a
workload close to the line rate. Thus, {\sys} can perform well within about 12-15\%
of the line rate.

\textbf{CPU utilization.}
\update{The CPU utilization is 3-10\% for the {\prepare} core and depends on the
DP shaping interval; the utilization is 8-70\% for the QUIC worker core,
which depends on the network I/O. The {\ushaper} core utilizes 100\% of the CPU as
it polls for packets from {\prepare}. As such, the {\prepare} and QUIC worker
cores would be able to support additional tunnel instances by time-sharing their
core. By using a polling
interval, we could reduce the CPU utilization of {\ushaper} to support
additional requests at the cost of additional latency.
In general, multiple tunnels can time-share the same physical cores, as long as
each core runs the same type of thread, to suffice property P4 mentioned in
\S\ref{sec:implementation}.
}

\subsection{Case Study: Video Streaming}
\label{subsec:eval-video}

\begin{figure}[t]
  \centering
  \begin{subfigure}[b]{0.49\columnwidth}
      \centering
      \includegraphics[width=\textwidth]{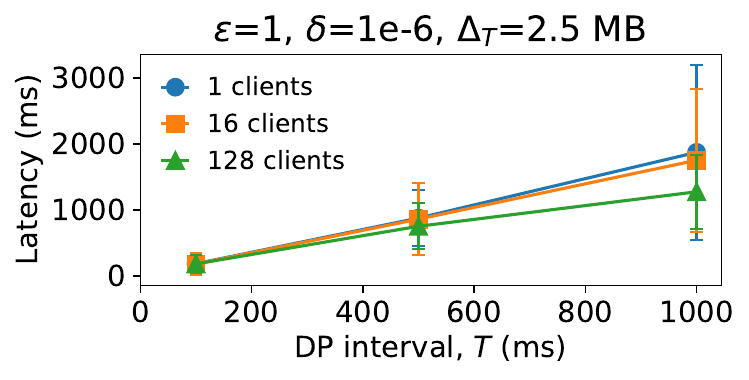}
      \caption{Latency vs DP interval}
      \label{fig:video-lat-vs-dpInt}
  \end{subfigure}
  \hfill
  \begin{subfigure}[b]{0.49\columnwidth}
      \centering
      \includegraphics[width=\textwidth]{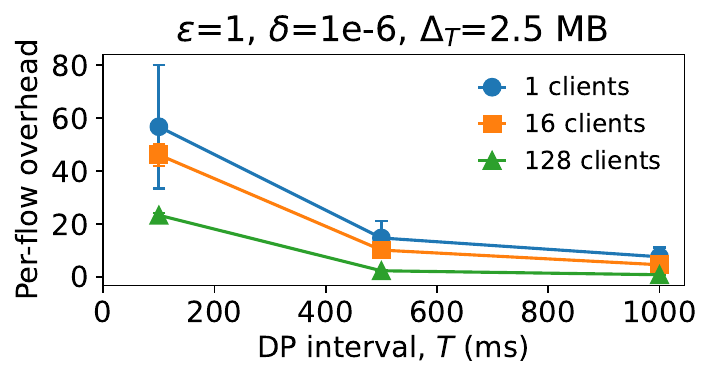}
      \caption{BW overhead vs DP interval}
      \label{fig:video-overhead-vs-dpInt}
  \end{subfigure}
  \begin{subfigure}[b]{0.49\columnwidth}
      \centering
      \includegraphics[width=\textwidth]{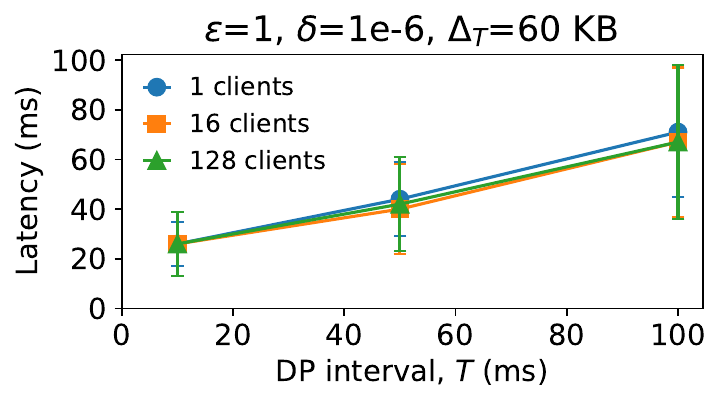}
      \caption{Latency vs DP interval}
      \label{fig:web-lat-vs-dpInt}
  \end{subfigure}
  \hfill
  \begin{subfigure}[b]{0.49\columnwidth}
      \centering
      \includegraphics[width=\textwidth]{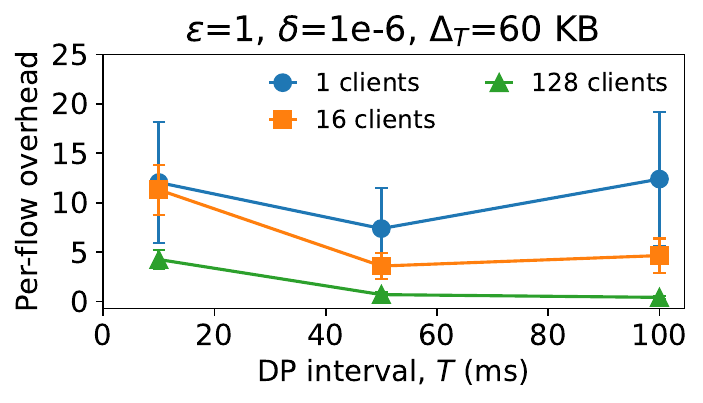}
      \caption{BW overhead vs DP interval}
      \label{fig:web-overhead-vs-dpInt}
  \end{subfigure}
  \caption{Latency and bandwidth overhead for different values of DP interval for video streaming (a, b) and for web (c, d).}
  \begin{subfigure}[b]{0.485\columnwidth}
      \centering
      \includegraphics[width=\textwidth]{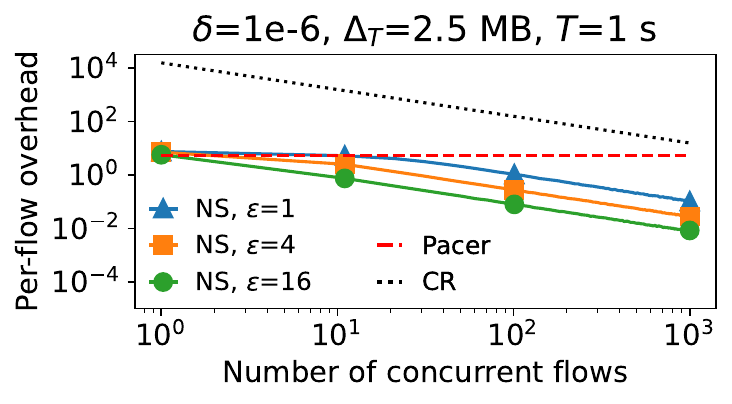}
      \caption{Video streaming.}
      \label{fig:video-overheads-compare}
  \end{subfigure}
  \hfill
  \begin{subfigure}[b]{0.485\columnwidth}
      \centering
      \includegraphics[width=\textwidth]{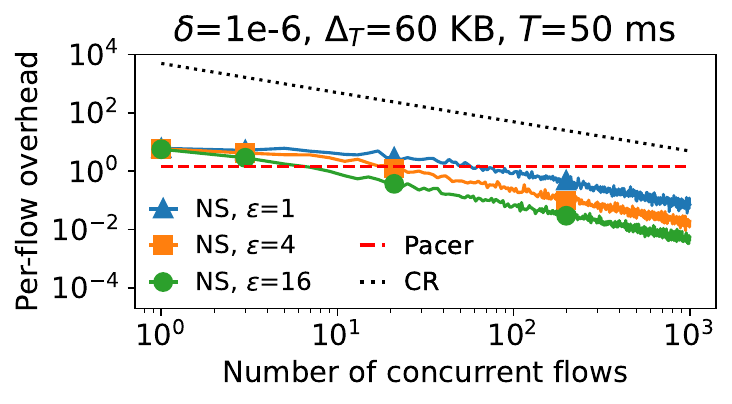}
      \caption{Web service.}
      \label{fig:web-overheads-compare}
  \end{subfigure}
  \caption{Comparison of {\sys}, constant shaping, Pacer.}
  \vspace{-0.4cm}
\end{figure}
Next, we examine the effect of different privacy settings on bandwidth and
latency overheads for video streaming clients.

We run experiments with three values of $\dpintvl$
for the server: 100ms, 500ms, and 1s, and max per-flow DP length cutoff of
1.21 MB, 1.22MB, and 1.7 MB, respectively. We use
\update{$\qsens = 2.5 MB$} and
\update{$\varepsilon = 1$}.
For all experiments, we set the DP parameters for client request traffic as
follows: \update{$\qsens = 200$}~bytes, $\winlen = 1s$, $\dpintvl = 10ms$,
\update{$\varepsilon = 1$} and max per-flow DP length cutoff of 206 bytes.
We run experiments with 1, 16, and 128 video
clients; each client requests one video randomly selected from the dataset.
For each set of configurations, we measure the average response latency for
individual video segments with the testbed as well as the per-flow relative
bandwidth overhead for the video streams in the simulator.

\textbf{Latency and bandwidth overhead.}
\Cref{fig:video-lat-vs-dpInt,fig:video-overhead-vs-dpInt} respectively show the
average segment download latency and the average per-flow relative bandwidth
overhead as a function of different intervals and for varying number of clients.
The {\base} segment download latency is \update{2.86 $\pm$ 1.41ms.} The latency
variance is due to variances in the segment sizes.
The relative bandwidth overhead of a video is the
number of dummy bytes transmitted normalized to the size of the unshaped video
stream.
The error bars show the standard deviation in latency and bandwidth overhead.

First, \Cref{fig:video-lat-vs-dpInt} shows that, for all values of $\dpintvl$,
the video segments can be downloaded well within 5s, which is the time to play
each segment and request the next segment from the server.
\update{The high variance is due to negative DP noise, which delays
payload transmission.}
Secondly, the results show the trade-off between latency and bandwidth. A larger
DP interval implies higher download latency but fewer queries, thus
yielding a lower bandwidth overhead.
Thirdly, with multiple concurrent clients, the bandwidth overhead is
amortized, while the average download latency only depends on $\dpintvl$.
Overall, {\sys} can secure video streams with low bandwidth overheads
and no impact on the streaming experience.

\subsection{Case Study: Web Service}\label{subsec:eval-web}
%
%
We perform similar experiments as \S\ref{subsec:eval-video} with our web
service. For the server responses, we use \update{$\qsens = 60 KB$}, $\winlen =
1s$, and \update{$\varepsilon = 1$}. We use $\dpintvl$ of {10ms, 50ms, and
100ms}, and per-flow DP length cutoffs of 60.8KB, 60.8KB, and 110.9KB,
respectively.
For the client requests, we use the same configs as in
\S\ref{subsec:eval-video}. We run \update{3 min} experiments with 1, 16, and 128
wrk2 clients with a total load of \update{1600 req/s};
each client requests random web pages from the dataset. We discard the numbers
of the first minute.

\textbf{Latency and bandwidth overhead.}
%
\Cref{fig:web-lat-vs-dpInt,fig:web-overhead-vs-dpInt} respectively show the
average response latency and the average per-web page relative bandwidth
overhead, across all web page requests.
The {\base} latency is \update{0.225 $\pm$ 0.3 ms}. (The high variance is due to
the time precision in wrk2 being restricted to 1ms.)
The web workload is more sporadic than video streaming, thus web page download
latencies have higher variance than video segment download latency.
The absolute latency overhead for {\ns} depends on the choice of $\dpintvl$. The
relative overhead depends on the underlying network latency, which unlike our
testbed, is in the order of 10s to 100s of milliseconds in the Internet.
%
Interestingly, the bandwidth overhead for web traffic first reduces with
increasing DP shaping interval from 10ms to 50ms, but then increases again
with an interval of 100ms. This is because, for small web pages, the DP interval
of 100ms is larger than the total time required to download web pages. As a
result, additional overhead is incurred due to the padding of traffic in the
100ms intervals.

\subsection{Comparison with other techniques}\label{subsec:eval-comparison}
%
%

\Cref{fig:video-overheads-compare,fig:web-overheads-compare} show the per-flow
relative bandwidth overhead of {\ns},
{\constshape}, and {\pacer} for video and web
applications, respectively, for varying number of concurrent flows.



For both video and web traffic, {\ns} incurs three orders of magnitude lower
overhead than {\constshape}, which requires continuously transmitting traffic at
the peak server load (configured for 1000 clients).
For video and web traffic, {\ns} requires 11 flows and more than 40 flows,
respectively to achieve lower overhead {\pacer}.
Pacer shapes server traffic only upon receiving a client request and does not
shape client traffic. Thus, it leaks the timing and shape of client requests,
which could potentially reveal information about the server responses
\cite{chen2010reality}.
{\sys} shapes traffic in both directions, which incurs higher overhead at the
cost of stronger privacy than Pacer.

\smallskip\noindent
\textbf{Evaluation summary.}
{
Our evaluation provides four insights.
(i) There is a huge gap between the theoretical DP guarantees and the
privacy configurations required to defeat SOTA attacks.
(ii) The latency overhead is dominated by the choice of DP shaping interval,
(iii) {\sys}'s middlebox can match about 88\% of the 10Gbps NIC line rate; a
single core of
{\ushaper} can match the peak throughput of a single core server,
(iv) {\sys}'s cost is in the two additional cores for {\prepare} and QUIC worker,
which helps to avoid any secret-dependent interference in shaping and keep low DP
shaping loop lengths. By optimising the implementation, we could use a
single middlebox to support larger workloads.
}

\vspace{-0.1cm}
\section{Related Work}
\label{sec:related}


\textbf{Traffic shaping for web.}
Prior work used traffic shaping~for defending against website
fingerprinting attacks.
%
Walkie-Talkie \cite{wang2017walkietalkie}, Supersequence
\noindent \cite{wang2014supersequence}, and Glove \cite{nithyanand2014glove} use
clustering techniques to group objects of a corpus and then shape the
traffic of all objects within each cluster to conform to a similar pattern.
Traffic morphing \cite{wright2009morphing} makes the traffic of one page look
like that of another.
These techniques compute traffic shapes that envelope or resemble the network
traces of individual objects. Hence, they require a large number of
traces to account for network variations.
{\sys} dynamically adapts traffic shapes based on the
prevailing network~conditions.

Cs-BuFLO~\cite{cai2014csbuflo}, Tamaraw~\cite{cai2014tamaraw}, and
DynaFlow~\cite{lu2018dynaflow}, determine traffic shape directly at runtime.
They pad object sizes to values that are correlated
with the original object sizes, such as the next
multiple or power of a configurable constant.
These defenses provide privacy akin to k-anonymity, but have no control on the
the size of the anonymous cluster.
Tamaraw\cite{cai2014tamaraw} formalizes the privacy guarantees.
{\sys}'s DP guarantees are strictly stronger than Tamaraw's (proof
in~\S\ref{appendix:tamaraw}).
%


\textbf{Differential privacy over streams.}
Dwork et al.~\cite{dwork2010dpcontinual} study DP on streams, in which
neighboring streams differ in at most one element and the query
counts over the stream prefix (without forgetting
old information) under a fixed DP budget for the entire stream.
{\sys} requires a stronger neighboring definition to model application data
streams (Def. \ref{def:neighboring-streams}), but is able to forget the past by
dropping stale data from our queue (\Cref{assumption:window}) and compose
privacy loss over time.

{\sys}'s neighboring definition is closer to that of user-level DP over streams
in \cite{dwork2010pan}. Instead of counting discrete change,
however, we use the L1 distance which enables coarsening.
Pan-privacy considers an adversary that can compromise the internals of
the algorithms (\eg our buffering queue). This makes the design of algorithms
challenging and costly.
Instead, we consider the buffering queue private and focus on
practical algorithms to study the privacy/overheads trade-off.


Kellaris et al.~\cite{kellaris2014differentially} introduce a notion of DP,
called $w$-event privacy, for streams of length $w$.
Neighboring stream pairs are those whose individual events are pairwise
neighbors within a window of upto $w$.
{\sys}'s neighboring definition accounts for the maximum stream distance over
{\em any} window of length $\winlen$, which is a better match for the streams
we consider.

Zhang et al.~\cite{zhang2019statistical} generate differentially private shapes
for video streams using Fourier Perturbation Algorithm (FPA)
\cite{rastogi2010differentially}.
FPA transforms a finite time series of bursts, into a series of DP shaped bursts
of the same length.
FPA requires the entire stream's profile upfront, and cannot guarantee complete
transmission of an input stream within the shaped trace.
{\sys}'s DP guarantees simply compose over burst interval sequences, thus
allowing shaping of streams of arbitrary lengths with
quantifiable privacy and overheads.



\textbf{Adversarial defenses.}
Adversarial defenses~\cite{rahman2020mockingbird, hou2020wf, abusnaina2020dfd,
shan2021dolos, nasr2021defeating, gong2022surakav}
generate targeted and low-overhead noise to defeat specific classifiers.
{\sys} provides provable and configurable privacy against both SOTA as
well as future classifiers.

\textbf{Network side-channel mitigation systems.}
{\sys}'s shaping tunnel is conceptually similar to Pacer's~\cite{mehta2022pacer}
cloaked tunnel.
However, Pacer mitigates leaks of a Cloud tenant's secrets to a
colocated adversary through contention at shared network links.
Pacer's cloaked tunnel controls the transmit time of
TCP packets in accordance with the shaping schedule and congestion control
signals. Thus, Pacer requires
non-trivial changes to the network stack on the end hosts.
%
{\sys} protects applications that are behind private networks but communicate
using the public Internet.
{\sys}'s tunnel endpoints can be placed at the interface of
the private-public networks, \eg in a middlebox, thus supporting
multiple applications without modifying end hosts.
Moreover, by shaping above the transport layer, {\sys} needs to control only the
precise timing for generation of bursts of DP length and not the subsequent
transmission to the network.


Ditto~\cite{meier2022ditto} and {\sys} propose shaping traffic at network nodes
separate from end hosts.
{\sys} proposes a hardware-independent, modular and portable
middlebox architecture that can be integrated with end hosts, routers,
gateways, or even programmable switches as in Ditto.


\textbf{Systems with other goals.}
Censorship circumvention systems~\cite{mohajeri2012skypemorph,
winter2013scramblesuit, barradas2017deltashaper,
barradas2020poking, rosen2021balboa}
rely on traffic obfuscation, scrambling, and transformations of a sensitive
application’s shape to that of a non-sensitive application. These techniques
prevent identification of original protocols by deep packet inspection,
but do not prevent inference of secrets from traffic shapes.
Karaoke~\cite{lazar2018karaoke} and Vuvuzela~\cite{van2015vuvuzela} are
anonymous messaging systems that use~DP to hide participants in a conversation,
but use constant-rate traffic among the participants.
\update{
AnoA~\cite{backes2013anoa} is a framework to analyze anonymity
properties of anonymous communication
protocols. AnoA supports DP based quantification for various
properties, such as sender anonymity and sender unlinkability.
}
{\sys}'s differentially-private traffic shaping hides the traffic content. In
principle, {\sys} could be combined with an anonymous communication system to
provide both content privacy and anonymity with~DP.


\vspace{-0.2cm}
\section{Conclusion}
\label{sec:conclusion}

{\sys} is a provably secure network side-channel mitigation system that provides
quantifiable and tunable privacy guarantees in traffic shaping.
\todo{We believe that {\sys} can eliminate the arms race in {\nsca} attacks and
defenses and can provide a portable and configurable framework for deploying
mitigations for applications with diverse traffic characteristics and in
different settings.}
{\sys}'s DP based traffic shaping strategy as well as its
modular and portable tunnel design can be extended to mitigate leaks
in multi-node systems, but we leave the details to future work.

\vspace{-0.2cm}
\section{Acknowledgments}
We thank the reviewers and our shepherd for their constructive
feedback.
This work was supported by the Natural Sciences and Engineering Research
Council of Canada (NSERC) [RGPIN-2021-02961,
DGDND-2021-02961], the Department of National Defense (DND) [MN3-011], a
Google Research Scholar award, the Digital Research Alliance of Canada, and AWS
(through UBC Cloud Innovation Center).

{
\footnotesize
\bibliographystyle{abbrv}
\bibliography{paper}
}

\appendix

\section{Attack classifiers}
\label{appendix:tcn}

{\small
\textbf{Beauty and Burst.}
The Beauty and the Burst classifier (BB) \cite{schuster2017beautyburst} is a CNN
(convolutional neural network) consisting of three convolution layers, a max
pooling layer, and two dense layers. We use a dropout of 0.5, 0.7, and 0.5
between the hidden layers of the network. We train the classifier with an Adam
optimizer, a categorical cross-entropy function, a learning rate of 0.01, with a
batch size of 64, and for 1000~epochs.
}

\begin{figure}[t]
    \centering
    \includegraphics[width=\columnwidth]{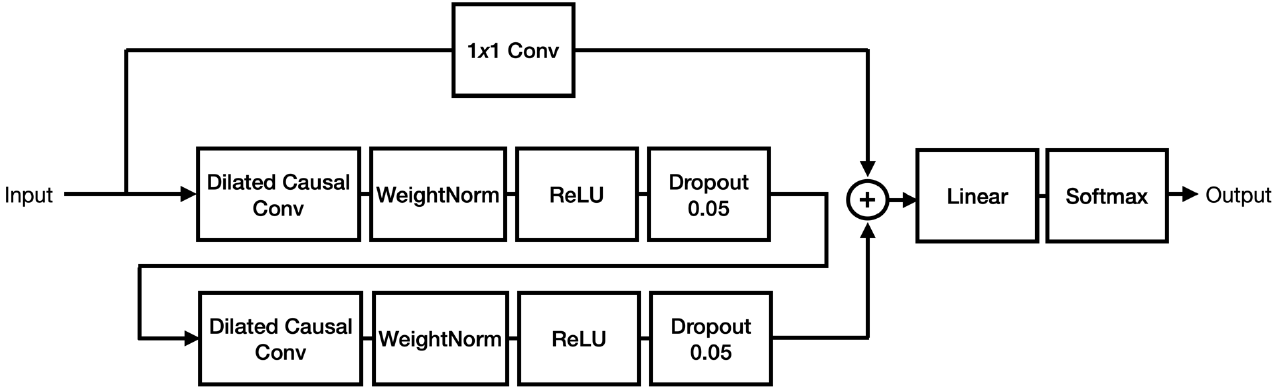}
    \caption{TCN classifier}
    \vspace{-0.4cm}
    \label{fig:tcn}
\end{figure}

{\small
\textbf{Temporal Convolution Network.}
%
While CNNs are generally effective in sequence modelling, they look at future
samples in a sequence and a very limited history of past samples to decide the
output of the current sample. Consequently, they require a large number of traces
and long traces for effective training and prediction.

Temporal Convolutional Networks (TCNs) \cite{bai2018tcn} overcome these problems
of CNNs by utilizing a one-dimensional fully-convolutional network equipped with
causal dilated convolutions, which allows them to examine deep into the past to
produce an output for the sequence at any given moment.

\Cref{fig:tcn} shows the architecture of our TCN classifier, which follows the
architecture proposed by Bai et al. \cite{bai2018tcn}. It
consists of two dilated causal convolutional layers, followed by weight
normalization and dropout layers with a dropout probability of 0.05. We train
the classifier for 1000 epochs.
}

\section{Extended evaluation of privacy vs overheads}
\label{sec:eval-extended}

\begin{figure}[t]
    \centering
    \begin{subfigure}{0.49\columnwidth}
      \includegraphics[width=\textwidth]{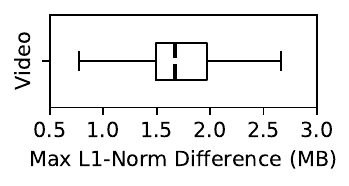}
    \end{subfigure}
    \hfill
    \begin{subfigure}{0.49\columnwidth}
      \includegraphics[width=\textwidth]{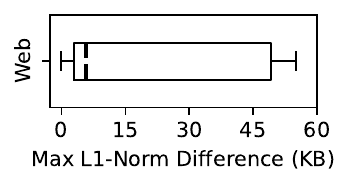}
    \end{subfigure}
    \caption{\update{Distribution of the difference in buffering queue lengths
    for application stream pairs.}
    }
    \vspace{-0.2cm}
    \label{fig:sensitivity-comparison}
\end{figure}

\begin{figure}[t]
    \centering
    \begin{subfigure}{0.49\columnwidth}
        \includegraphics[width=\textwidth]{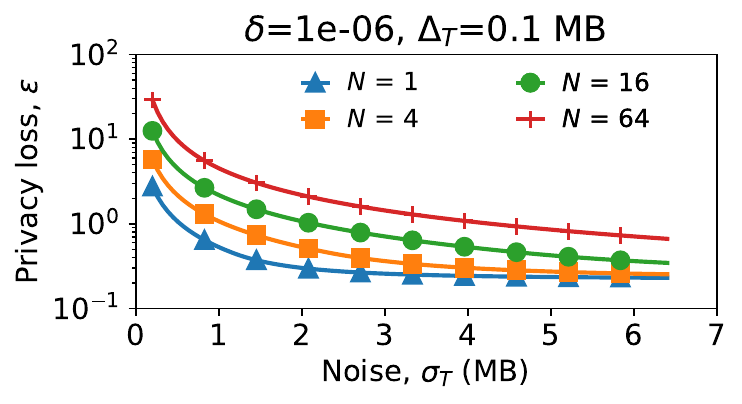}
        \caption{Noise vs privacy loss}
        \label{subfig:low-sens-epsilon-sigma}
    \end{subfigure}
    \hfill
    \begin{subfigure}{0.49\columnwidth}
        \includegraphics[width=\textwidth]{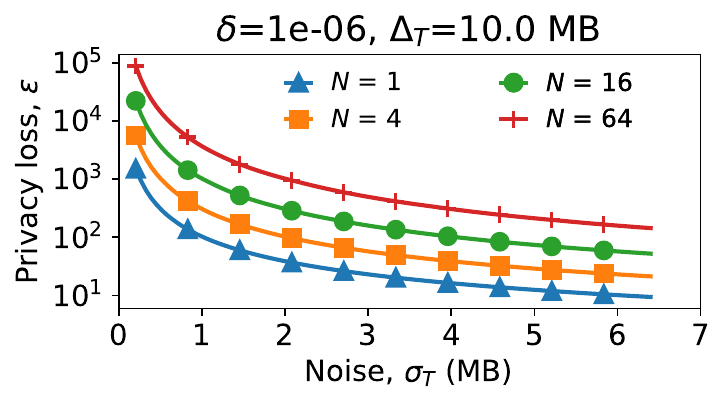}
        \caption{Noise vs privacy loss}
        \label{subfig:high-sens-epsilon-sigma}
    \end{subfigure}
    \hfill
    %
\begin{subfigure}{0.49\columnwidth}
    \includegraphics[width=\textwidth]{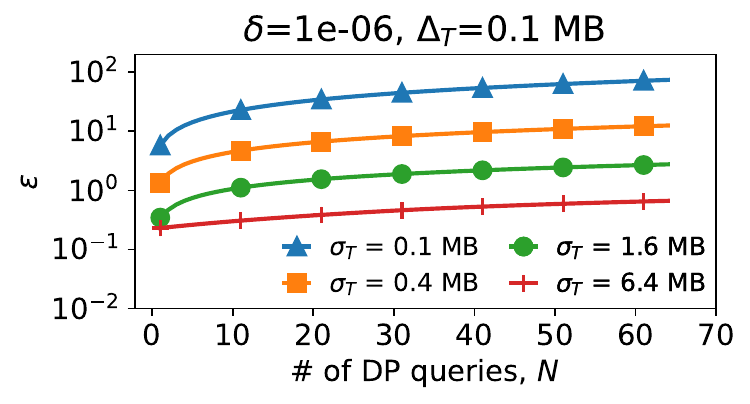}
    \caption{\# DP queries vs privacy loss}
    \label{subfig:low-sens-epsilon-queries}
\end{subfigure}
\hfill
\begin{subfigure}{0.49\columnwidth}
    \includegraphics[width=\textwidth]{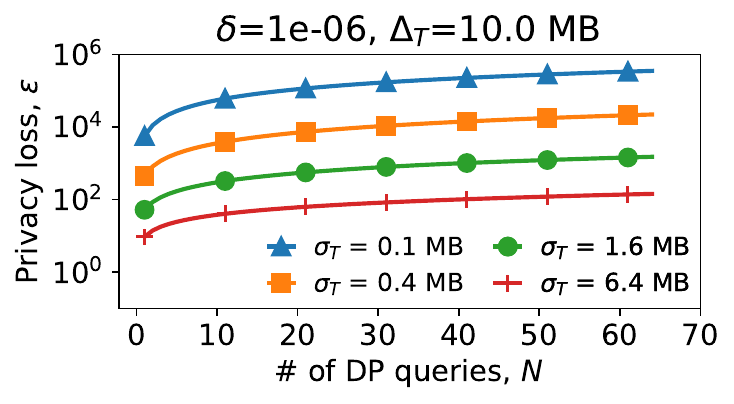}
    \caption{\# DP queries vs privacy loss}
    \label{subfig:high-sens-epsilon-queries}
\end{subfigure}
    %
    %
\caption{
    \update{Privacy loss vs noise and \# queries for different~$\qsens$.}
}
\vspace{-0.4cm}
\label{fig:privacy-microbenchmarks-extended}
\end{figure}

{\small
\textbf{Distribution of queue length differences.}
\Cref{fig:sensitivity-comparison} shows the distribution of the
buffering queue length differences generated for each pair of an
application's streams. We use $\winlen$ of 5s
for video streams and 1s for web pages.
The dashed lines show the median, which is \update{1.63 MB} and
\update{6 KB} for videos and web pages, respectively.

\textbf{Privacy loss vs overheads for other \update{$\qsens$}.}
\Cref{fig:privacy-microbenchmarks-extended} shows similar results as
\Cref{fig:privacy-microbenchmarks} for \update{$\qsens$} of 0.1 MB and 10 MB.
}

{\small
\textbf{Bandwidth overheads with a fixed cutoff.}
In \S\ref{sec:eval}, we explained that {\sys} applies a cutoff to the DP
shaped buffer length (if the sampled noise is very large) that depends on the
number of active flows. This reveals the number of active flows at any given
time.
To hide the number of flows, we can set
the cutoff to a fixed~value (\eg based on the
maximum flows that the system must support).

\Cref{subfig:bw-vs-dpInt-video-no-max} presents results similar to those of
\Cref{fig:video-overhead-vs-dpInt}, but with the cutoff for the shaped buffer
length set to a fixed value of 1.21 GB, 1.22 GB, and 1.7 GB.
\Cref{subfig:bw-vs-dpInt-web-no-max} presents results similar to those of
\Cref{fig:web-overhead-vs-dpInt} but with cutoffs of 60.8~MB, 60.8~MB, and 110.9
MB.
\Cref{subfig:bw-comparison-video-no-max,subfig:bw-comparison-web-no-max} show
the results similar to those of
\Cref{fig:video-overheads-compare,fig:web-overheads-compare}.
%
The fixed cutoffs correspond to 1000 flows, which lead to a significant increase
in {\sys}'s overheads.
However, the overheads quickly amortize with several concurrent flows.
}

\begin{figure}[t]
    \centering
    \begin{subfigure}{0.49\columnwidth}
        \centering
        \includegraphics[width=\textwidth]{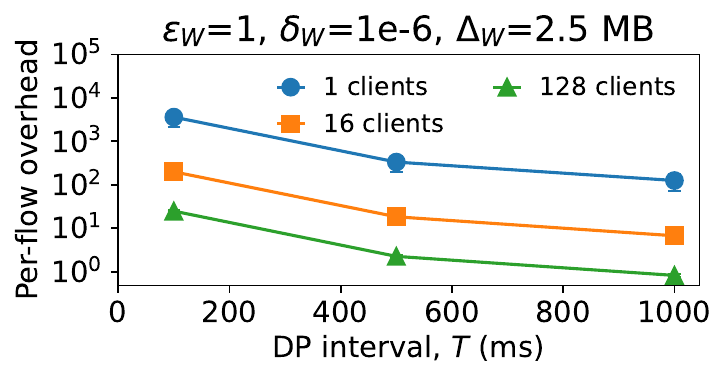}
        \caption{Video streaming}
        \label{subfig:bw-vs-dpInt-video-no-max}
    \end{subfigure}
    \hfill
        %
    \begin{subfigure}{0.49\columnwidth}
        \centering
        \includegraphics[width=\textwidth]{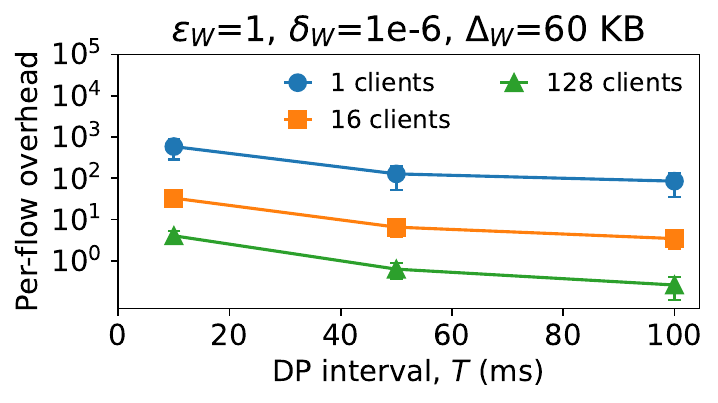}
        \caption{Web service}
        \label{subfig:bw-vs-dpInt-web-no-max}
    \end{subfigure}

    \centering
    \begin{subfigure}{0.49\columnwidth}
        \centering
        \includegraphics[width=\textwidth]{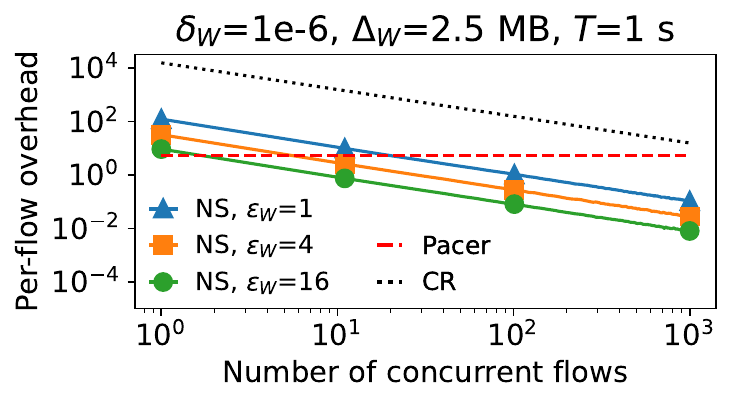}
        \caption{Video streaming}
        \label{subfig:bw-comparison-video-no-max}
    \end{subfigure}
    \hfill
        %
    \begin{subfigure}{0.49\columnwidth}
        \centering
        \includegraphics[width=\textwidth]{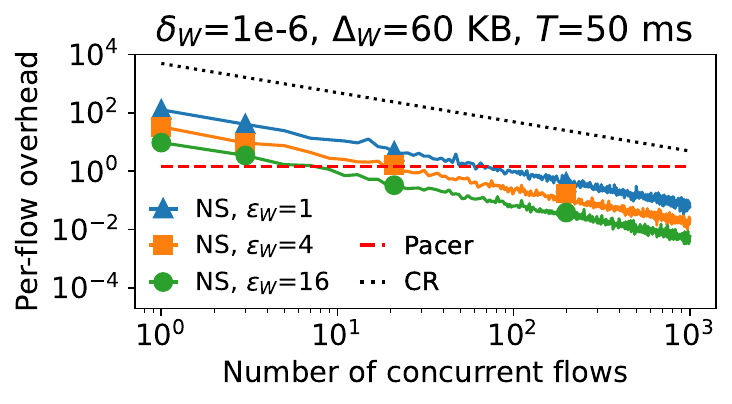}
        \caption{Web service}
        \label{subfig:bw-comparison-web-no-max}
    \end{subfigure}
    \caption{
        \update{Bandwidth overheads with a fixed cutoff.}
    }
\vspace{-0.4cm}
\end{figure}

\section{Proof of {\sys}'s DP based Shaping}
\label{appendix:dp}





{\small
\begin{numprop}[\ref{prop:sensitivity}]
  {$\sys$} enforces $\qsens \leq \ssens$.
\end{numprop}
\noindent To prove \Cref{prop:sensitivity}, we show that at any DP query time
$k$, any neighboring streams $\istream, \istream'$ can only create a queue size
difference such that $|\qlent{k} - \qlent{k}'| \leq \ssens$. This is formalized
in the following Lemma:
\begin{lemma}\label{lemma:sensitivity-time}
Assume two neighboring traffic streams $\istream$ and $\istream'$
($\|\istream-\istream'\|_1 \leq \ssens$), are transmitted through {\sys} and
get the same randomness draws $z_k$ (they are parallel worlds with an identical
  \sys, but the streams differ).
Then, at any DP query time $k$, the length of the buffering queue for
  $\istream$ and $\istream'$ are $\ssens$-close.
Formally:

\begin{equation}\label{equ:composition}
        \forall k \geq 0 : |\qlent{k} - \qlent{k}'| \leq \ssens
\end{equation}
\end{lemma}
\begin{proof}
{\sys} performs a DP query for the size of the buffering queue at intervals of
length $T$.
While transmitting $\istream$, the queue length at the end of the
$k$\textsuperscript{th} interval $\dpintvl_{k}$ is a function of three
variables:
%
(i) The buffering queue length $\qlent{k-1}$ at the end of
$\dpintvl_{k-1}$.
(ii) The total number of payload bytes that have been dequeued from the
buffering queue in the $k$\textsuperscript{th} interval,
${\payload}_{k}$.
(iii) The number of new payload bytes from the application stream added
to the buffering queue since the previous interval, which is the sum of
sizes of all packets arriving between $(k-1)$\textsuperscript{th} and
$k$\textsuperscript{th} interval, \ie
$\sum_{\dpintvl_{k-1} \leq t < \dpintvl_{k}} P^S_t$.
%
%
Therefore, the buffering queue length after dequeue in $\dpintvl_{k}$ is:
\begin{equation}\label{equ:queue-state}
\qlent{k} = {\qlent{k-1}} +
{\sum_{\dpintvl_{k-1} \leq t < \dpintvl_{k}} P^S_t}
-
{{\payload}_{k}}
\end{equation}
The difference between the queue lengths of
$\istream, \istream'$ at $\dpintvl_{k}$ is:
{\footnotesize
\begin{equation}\label{equ:queue_state_expansion}
\qlent{k} - \qlent{k}'~=~(\qlent{k-1} - \qlent{k-1}')~+ (\sum_{\dpintvl_{k-1} \leq t < \dpintvl_{k}}P^S_t - \sum_{\dpintvl_{k-1} \leq t < \dpintvl_{k}}P^{S'}_t) - ({\payload}_{k} - {\payload}_{k}')
\end{equation}}
We divide the proof into two steps.
First, we show that the dequeue stage of the shaping mechanism does not increase the difference in queue lengths.
Secondly, we show that under \Cref{assumption:window}, the enqueue stage of incoming streams does not increase the difference in queue lengths beyond $\ssens$.


To show that the dequeue stage never increases the difference between queue
lengths, we show that in any period $k$, we always dequeue more (or equal) data from
the larger queue at the time of the DP query. Formally:
\vspace{-0.2cm}
\begin{equation}\label{equ:queue-dequeue}
        (\qlent{k-1} - \qlent{k-1}') \cdot ({\payload}_{k} - {\payload}_{k}') \geq 0 ,
\end{equation}
where the indexes are because removal amount in period $k$ depends on the
query result for the queue at $k-1$.
Assume without loss of generality that $\qlent{k-1} \geq \qlent{k-1}'$.
Since each DP query receives the same noise, we get $\qlendpt{k-1} \geq
\qlendpt{k-1}'$,
and thus:
\[
  {\payload}_{k} = \min\big\{0, \qlendpt{k-1}, \qlent{k-1} \big\} \geq
\min\big\{0, \qlendpt{k-1}', \qlent{k-1}'\big\} = {\payload}_{k}' .
\]
The case for $\qlent{k-1} \leq \qlent{k-1}'$ is symmetric, and we have
 $\textrm{sign}(\qlent{k} - \qlent{k}') = \textrm{sign}({\payload}_{k} - {\payload}_{k}')$,
which implies \Cref{equ:queue-dequeue}.
Moreover, we show:
\begin{equation}\label{equation:deuque-size}
\vspace{-0.1cm}
  |{\payload}_{k} - {\payload}_{k}'| \leq |\qlent{k} - \qlent{k}'| ,
\end{equation}
using two cases and assuming that $\qlent{k} \geq \qlent{k}'$ (again the
opposite case is
symmetric). Either the DP noise is $\geq -\qlent{k}'$, and $\qlent{k} -
\qlent{k}' = {\payload}_{k} - {\payload}_{k}'$. Or the DP noise is  $< -
\qlent{k}'$, in which case ${\payload}_{k}' = 0$ but ${\payload}_{k} \leq
\qlent{k} - \qlent{k}'$. Either way, \Cref{equation:deuque-size} holds.
We can now study the queue length difference over time:
\begin{align*}
\nonumber
|\qlent{k} - \qlent{k}'| & =  |(\qlent{k-1} - \qlent{k-1}')
+
(\sum_{\dpintvl_{k-1} \leq t < \dpintvl_{k}}P^S_t - P^{S'}_t) - ({\payload}_{k} -
{\payload}_{k}') |
\\
& \leq
|(\qlent{k-1} - \qlent{k-1}') - ({\payload}_{k} - {\payload}_{k}')|
+
\sum_{\dpintvl_{k-1} \leq t < \dpintvl_{k}}|P^S_t - P^{S'}_t|
\\
& \leq
|\qlent{k-1} - \qlent{k-1}'|
+
\sum_{\dpintvl_{k-1} \leq t < \dpintvl_{k}}|P^S_t - P^{S'}_t|
\end{align*}
where the first equality uses \Cref{equ:queue_state_expansion} and the last
inequality uses \Cref{equ:queue-dequeue} and \Cref{equation:deuque-size}.
Intuitively, the dequeue stage never increases the difference between queue
lengths.

Finally, we are ready to bound the difference in queue length due to enqueues.
Let $d_k = |\qlent{k}-\qlent{k}'|$, and~$d_0 = 0$:
\begin{align}
\nonumber
  d_k &\leq d_{k-1}~+~\sum_{\dpintvl_{k-1} \leq t < \dpintvl_{k}}|P^S_t - P^{S'}_t|~=~
\nonumber
0 + \sum_{i=0}^{k}\big({\sum_{\dpintvl_{i-1} \leq t < \dpintvl_{i}}|P^S_t - P^{S'}_t|}\big)\\
&=
\nonumber
\sum_{0 \leq t < \dpintvl_{k}} |P^S_t - P^{S'}_t| =
\nonumber
\sum_{0 \leq t < \dpintvl_k - W} |P^S_t - P^{S'}_t| +
\sum_{\dpintvl_{k} - W \leq t < \dpintvl_{k}} |P^S_t - P^{S'}_t|\\
\nonumber
&\leq \ssens
\end{align}
since $\sum_{0 \leq t < \dpintvl_k - W} |P^S_t - P^{S'}_t| = 0$ by \Cref{assumption:window},
and $\sum_{\dpintvl_{k} - W \leq t < \dpintvl_{k}} |P^S_t - P^{S'}_t| \leq \ssens$ by \Cref{def:neighboring-streams}.
\end{proof}
\vspace{-0.1cm}
\Cref{lemma:sensitivity-time} $\implies$ \Cref{prop:sensitivity}, since
$\qsens = \max_{k}~\max_{\istream, \istream'} | \qlent{k} - \qlent{k}' | \leq
\ssens$.
}

\section{Comparison of {\sys} and Tamaraw}
\label{appendix:tamaraw}
{\small
Tamaraw \cite{cai2014tamaraw} provides a
mathematical notion of privacy guarantee of a shaping strategy, called
$\epsilon$-security.
To disambiguate with {\sys}'s notion of
$(\varepsilon, \delta)$-DP, we rename Tamaraw's $\epsilon$
variable with $\gamma$.
We show that {\sys}'s $(\varepsilon, \delta)$-DP definition
is strictly stronger than Tamaraw's $\gamma$-security definition.

First, we explain Tamaraw's $\gamma$-security definition.
In Tamaraw, $W$ is a random variable representing the label of a
traffic trace.
For a trace $w$, the random variables $T_{w}$ and $T_{w}^{D}$ respectively
represent the packet trace of $w$ before and after shaping with a defense $D$.
The distribution of $T_{w}^D$ captures all variations in the observed
shape of $w$ resulting from both the defense mechanism and the
network, while the distribution of $T_{w}$ only contains variations due to
the network.
The attacker can measure the distribution of $W$ and $T_{w}^{D}$ independently.
Upon observing a network trace $t$, an optimal attack~$A$,~selects the
label of the trace with maximum likelihood of observation:
\begin{equation*}
  A(t) = \argmax_{w}{\Pr[W=w]\Pr[T_{w}^{D}=t]}
\end{equation*}
The probability that an attack $A$ outputs label $w_i$ is $\Pr_A[w_i]$.
\begin{definition*}[Tamaraw $\gamma$-privacy]
  A fingerprinting defense $D$ is said to be uniformly $\gamma$-private if for the attack $\mathcal{A}$ if we have:
  \begin{equation*}
    \max_w\big[\Pr[A(T_w^D)=w]\big] \leq \gamma
  \end{equation*}
\end{definition*}
\begin{proposition}
  Tamaraw $\gamma$-privacy is strictly weaker than the notion of
  $\varepsilon$-differential privacy.
  \label{prop:tamaraw-vs-netshaper}
\end{proposition}

To prove \Cref{prop:tamaraw-vs-netshaper}, we prove the following two lemmas.
\begin{lemma}
  There exists a Tamaraw $\gamma$-private defense mechanism that fails to
  satisfy $\varepsilon$-DP for any given value of $\varepsilon$.
\end{lemma}
\begin{proof}
Consider a web service with a dataset of $n$ web pages.
We propose a defense $D$, which shapes pages to a constant-rate pattern $O_c$
with probabilities $\alpha$ or $\beta$.
$D$ reshapes page $w_j$ to $O_c$ with probability $\beta$ and
all other pages $w_i \neq w_j$ with probability $\alpha$.
If a page is not shaped, it is revealed.


The probability that any attack can correctly identify the label for $w_j$ is
upper-bounded by:  $\Pr[A(T^{D}_{w_{i=j}}) = w_j]$
\begin{align*}
  & = \Pr[A(T^D_{w_{i=j}}) = w_j |
  T^D_{w_{i=j}}=T_{w_{i=j}}]\Pr[T^D_{w_{i=j}}=T_{w_{i=j}}] +
  \\
  &~~~~\Pr[A(T^D_{w_{i=j}}) = w_j | T^D_{w_{i=j}}=O_c]\Pr[T^D_{w_{i=j}}=O_c]
  \\
  & \leq  1.(1-\beta) + \frac{1}{n}\beta = p_c^j
\end{align*}
For $(1- \frac{n\gamma - 1}{n-1}) < \beta$ we have: $p_c^j \leq \gamma$.
Similarly, the probability that any attack can correctly classify $w_{i\neq j}$
is upper-bounded by $p_c^i = 1-\alpha + \frac{\alpha}{n}$, and for $(1-
\frac{n\gamma - 1}{n-1}) < \alpha$ we have: $p_c^j \leq \gamma$.
Therefore, for all values of $\alpha$ and $\beta$ such that $(1- \frac{n\gamma -
1}{n-1}) < \beta < \alpha$, the probability that any attack can successfully
guess~a~page is less than $\gamma$, and
the defense is uniformly $\gamma$-private.

When the output of the algorithm is $O_{c}$, the probability of the page being
$w_j$ is $\beta$ and any other page is $\alpha$. Thus,
$\log(\frac{\Pr[T_{w_{i\neq j}}^{D}=O_{c}]}{\Pr[T_{w_{i=j}}^{D}=O_{c}]})
= \log(\frac{\alpha}{\beta})$.
By setting $\alpha > {e^{\varepsilon}}\beta$, we get a mechanism that is
$\gamma$-private for Tamaraw but not $\varepsilon$-DP for {\sys}.
\vspace{-0.2cm}
\end{proof}
\begin{lemma}
  A $\varepsilon$-DP shaping algorithm is Tamaraw $\gamma$-private for
  $\varepsilon \leq \log(n\gamma)$.
\end{lemma}
\begin{proof}
For a trace $w$, the random variable $T_{w}^{DP}$ represents the packet
trace of $w$ after a differentially private shaping mechanism is applied.
The classification attack on shaped traffic can
be considered a post-processing of the results of a
differentially private shaping mechanism (i.e. defense) and, hence, is
differentially
private. Therefore, we have:
\begin{align*}
& \frac{\Pr[A(T_{w_{i}}^{DP}) = w_i]}{\Pr[A(T_{w_{j}}^{DP}) = w_i]} \leq e^
{\varepsilon}
\Rightarrow \Pr[A(T_{w_{i}}^{DP}) = w_i] \leq e^
{\varepsilon} .\Pr[A(T_{w_{j}}^{DP}) = w_i]
\end{align*}
Intuitively, this implies that the likelihood of the attacker correctly
classifying the trace with label $i$ compared to incorrectly classifying it with
label $j$ is bounded by $e^{\varepsilon}$.
The above inequality is correct for all $w_j: j\in \{1, 2, \dots, n\}$, and we
can extend the above equation to calculate the summation over $j$:
\vspace{-0.3cm}
\begin{align*}
n\times \Pr[A(T_{w_{i}}^{DP}) = w_i] \leq e^{\varepsilon}\sum_{j=1}^{n}
\Pr[A(T_{w_{j}}^{DP}) = w_i]
= e^{\varepsilon} \operatorname{Pr}_{A}[w_i].
\end{align*}
Thus, for any given trace $w_i$, the probability that any attack $A$,
classifies it correctly is bounded by:
$\Pr[A(T_{w_{i}}^{DP}) = w_i] \leq \frac{e^{\varepsilon} \Pr_{A}[w_i]}{n}$.
The probability that an attacker can guess the victim’s trace is bounded by:
$\max_{w_i}{\Pr[A(T_{w_{i}}^{DP}) = w_i]} \leq \frac{e^{\varepsilon}}{n}
\max_{w_i}{\operatorname{Pr}_{A}[w_i]} \leq \frac{e^{\varepsilon}}{n} \leq
\gamma$.
\end{proof}
}

\end{document}